\documentclass[nonblindrev]{informs3}

\usepackage{natbib}
 \bibpunct[, ]{(}{)}{,}{a}{}{,}%
 \def\bibfont{\small}%

\usepackage{geometry}
\usepackage{booktabs}
\usepackage[table, dvipsnames]{xcolor}
\usepackage{xargs}
\usepackage{framed}

\usepackage{tikz} 

\usepackage[linesnumbered,ruled,vlined]{algorithm2e}
\usepackage{float}

\usepackage{natbib} 
 \bibpunct[, ]{(}{)}{,}{a}{}{,}%
 \def\bibfont{\small}%
\usepackage{multirow}
\usepackage{caption}
\usepackage{longtable}
\usepackage{supertabular}

\usepackage{colortbl}
\usepackage[norefpage]{nomencl}
\usepackage{enumitem}

\usepackage{graphicx}
\usepackage{graphics}
\usepackage{Input}
\usepackage{qtree}
\usepackage{amsmath}
\usepackage{amsfonts,latexsym,amssymb}
\usepackage{ifthen}
\usepackage[gen]{eurosym}
\usepackage{colortbl}
\usepackage{wasysym}
\usepackage{multirow}
\usepackage[caption=false]{subfig}

\usepackage[justification=centering]{caption} 

\usepackage{dsfont} 

\usepackage[colorinlistoftodos,prependcaption,textsize=small]{todonotes}

\newcommand{\MB}[1]{\textcolor{black}{#1}}

\newcommand{\el}[1]{\textcolor{green}{Eleni: #1}}

\newboolean{includeHidden}
\setboolean{includeHidden}{false}

\newcommand{\hide}[1]{\ifthenelse{\boolean{includeHidden}}{{\tiny\textbf{HIDDEN:~}#1}}{}}
\makeatletter
\def\munderbar#1{\underline{\sbox\tw@{$#1$}\dp\tw@\z@\box\tw@}}
\makeatother

\newcommand{\mcal}[1]{\ensuremath{\mathcal{#1}}}

\DeclareMathOperator{\R}{\mathbb{R}}
\DeclareMathOperator{\Z}{\mathbb{Z}}

\TheoremsNumberedThrough     

\EquationsNumberedThrough    

\MANUSCRIPTNO{}

\begin{document}
\RUNAUTHOR{tbd}

\RUNTITLE{Assignment Markets with Budget Constraints}

\TITLE{Assignment Markets with Budget Constraints}
\ARTICLEAUTHORS{%
\AUTHOR{Eleni Batziou, Martin Bichler, Maximilian Fichtl\footnote{We thank Itai Ashlagi, Shahar Dobzinski, Kassian Koeck, Paul Milgrom, Alexander Teytelboym, Rakesh Vohra, Zaifu Yang, seminar and conference participants at the SLMath fall program on market and mechanism design in 2023, the ACM Conference on Economics and Computation 2022, the AdONE Retreat 2021, and the Dagstuhl Workshop on Matching under Preferences 2021 for helpful comments and suggestions. Eleni Batziou acknowledges the support of the German Research Foundation (DFG) within the Research Training Group AdONE (GRK 2201). Martin Bichler is grateful for support by the German Research Foundation (DFG) (BI 1057/1-9).}}
\AFF{School of Computation, Information, and Technology, Technical University of Munich, 85748 Garching, Germany}

} 

\ABSTRACT{%
This paper studies markets where a set of indivisible items is sold to bidders with quasilinear, unit-demand valuations, subject to a hard budget constraint. Without financial constraints the well-known assignment market model of Shapley and Shubik (1971) allows for a simple ascending auction format that is incentive-compatible, and strongly Pareto-optimal. However, this auction model does not capture the possibility that bidders face hard budget constraints. 
We design an iterative auction that depends on demand queries and an easily verifiable additional condition to maintain the properties in the presence of budget constraints.  
If instead this additional condition does not hold, incentive compatibility and core stability are at odds, and we cannot hope to achieve strong Pareto optimality in a simple ascending auction even with truthful bidding.
Moreover, even in a complete information model where the auctioneer has access to valuations and budget constraints, the problem is NP-hard.\\

\footnote{An earlier version of this paper was circulated under the title: ’Core-Stability in Assignment Markets with Financially Constrained Buyers’}
}

\maketitle
{\begin{center} Version \today \end{center}}


\section{Introduction}


The assignment market model of \citet{shapley1971assignment} is a fundamental model in economics. Bidders have quasilinear preferences and they demand at most one out of several items. A housing market where each buyer is interested in at most one item is a classic example. 
\citet{shapley1971assignment} show that in such a market with quasilinear and transferable utility,  maximization of aggregate surplus (the sum of consumer and producer surplus) is a linear optimization problem and the dual variables constitute competitive equilibrium prices. Such an outcome is in the strong core because any deviation from this equilibrium would reduce the total surplus and thus leave at least one agent worse off.
Subsequent contributions showed that the lowest prices in the core are incentive-compatible for buyers \citep{leonard1983elicitation} and that these Walrasian prices can be found in polynomial time with a simple ascending auction implementing a variant of the Hungarian algorithm \citep{demange1986multi}. In summary, such an ascending auction is incentive-compatible, budget-balanced, strongly Pareto-efficient, and it can be run efficiently. 

We ask if these desiderata are conflicting when buyers have a hard budget constraint, i.e., a limit on how much they are allowed to spend in an auction.\footnote{A soft budget constraint refers to the case where an agent might get additional liquidity at some cost, which our model does not allow \citep{baldwin2023equilibrium}.}  
Quasilinear utilities with hard, binding budget constraints have been studied in many markets including spectrum auctions \citep{Bichler16, janssen_karamychev_kasberger_2017}, display ad or sponsored search auctions \citep{wu2018budget, conitzer2017multiplicative, borgs2007dynamics, dobzinski2014efficiency, conitzer2022pacing}. They have also been discussed in the literature on optimal auction design \citep{laffont1996optimal, che1998standard, benoit2001multiple, pai2014optimal}. 
Imperfect capital markets were mentioned as a reason for such budget constraints \citep{che1998standard, che2013assigning}, but also delegation problems within bidding firms \citep{bichler2018principal}. \citet{pai2014optimal} write that ``not every potential buyer of a David painting who values it at a million dollars has access to a million dollars to make the bid.'' One can argue that there is always some limit on how much an agent can spend in the market and the ability-to-pay is not infinity, but that this constraint is not necessarily binding, i.e., below the value that an agent has for a good. 
While a quasilinear objective with a hard budget constraint is a stark model of preferences, it is a good approximation of environments in the field and it is widely used in the literature. 

We stick to bidders with quasilinear preferences with unit demand as in the assignment model of \citet{shapley1971assignment}, but add hard and binding budget constraints. We are interested in deterministic mechanisms where money can be used to transfer utility up to a certain limit on spending. 
We look for mechanisms that implement strongly Pareto-efficient outcomes in the core. 
With budget constraints, the core and the strong core no longer coincide as is the case where budgets are so high that they do not constrain the transfers. In a core solution there can be possibilities for Pareto improvements, as the following example shows.

\begin{example}\label{ex:example1}
Suppose that an item is being sold and agent 1 has a value of \$1, while agent 2 has a value of \$100. Both agents have no budget constraint. This is the standard quasilinear case with transferable utility. Now, assigning the item to agent 2 at a price of \$1 is a competitive equilibrium, which is a strongly Pareto-optimal solution. 
If we add a budget constraint of \$1 for both agents, there is no reason to change the allocation to agent 1 with a value of only \$1. Suppose the item was originally assigned to agent 1, then reassigning the item to agent 2 constitutes a Pareto improvement, making the outcome strongly Pareto-optimal. The same allocation is still a core outcome that is strongly Pareto optimal, but the reverse allocation of the item to agent 1 is not.
The utility of agent 1 at a price of \$1 is zero, independent of the allocation. 
\end{example}
 

\subsection{Contributions}

First, we show that in the assignment model with binding budget constraints, selecting a core outcome that maximizes aggregate surplus is a sufficient condition for strong Pareto optimality. 
Then we introduce and analyze an iterative auction that always finds an outcome in the core with bidders who maximize utility in each round and have unit demand and a hard budget constraint. 
Our algorithm relies exclusively on demand queries via prices (no value or budget queries) and provides a natural generalization of the \textit{ascending auction} by \citet{demange1986multi} and ascending auctions as they can be found in the field.
We identify an additional condition on top of the unit-demand valuations when this auction leads to a unique solution that is strongly Pareto-optimal. Under this condition, at most one bidder hits his budget constraint in a round. In this case, the auction is ex-post incentive-compatible and maintains all properties of the auction by \citet{demange1986multi}. However, we also show that when this additional condition is not met, then no ex-post incentive-compatible mechanism can terminate in a solution in the core for every input. 


\MB{When multiple bidders hit their budget constraint in a round, the auctioneer can exclude one of these bidders and restrict her to bid on certain items in subsequent rounds. Then, the auction always finds an outcome in the core, but it does not necessarily find one that is strongly Pareto-efficient. 
For any instance of our iterative auction, if the auctioneer would be able to guess the right decisions throughout the auction, the auction terminates in a strongly Pareto-efficient solution. Otherwise, aggregate surplus can be arbitrarily low and there can be a possibility for Pareto improvements. However, guessing the right decision is hard. Even if the auctioneer, whenever a decision has to be made regarding which bidder to exclude, had access to the true valuations of bidders who hit their budget constraint, he could not determine the right bidder to exclude with certainty. }
We show that even if the auctioneer had complete information about valuations and budgets of all bidders, then finding a surplus-maximizing outcome in the core is NP-hard and no simple (polynomial-time) iterative auction can find this solution. So not only is the lack of incentive compatibility a concern, there are also fundamental computational barriers to finding a strongly Pareto-optimal outcome in the core, if some bidders experience binding budget constraints.


\subsection{Organization}

In Section 2, we discuss related literature, before we introduce notation and relevant definitions in Section 3. Section 4 briefly summarizes assignment markets without binding budget constraints, and Section 5 discusses the differences that arise when budget constraints become binding. Section 6 introduces an iterative auction and describes its properties, and Section 7 shows that computing an outcome that is always strongly Pareto-optimal is NP-hard. Section 8 provides conclusions.

\hide{
\subsection{An Illustrative Example}

In order to illustrate the problem, we introduce a simple example with two buyers and two sellers (see Figure \ref{fig:example}). Without budget constraints an auctioneer could implement an outcome that is in the core and maximizes aggregate surplus using the ascending auction by \citet{demange1986multi}. With budget constraints there is no anonymous and linear price vector, for which the surplus-maximizing allocation (ignoring budget constraints) would be such that no pair of buyers and sellers would want to deviate. 

In more detail, buyer $B_1$ has a value of \$4 for the item of seller $S_1$ and of \$10 for the item of seller $S_2$. Buyer $B_2$ has a value of \$3 for the item of seller $S_1$ and of \$6 for the item of seller $S_2$. The budget constraints of $B_1, B_2$ are $b_1 = 2$ and $b_2 = 4$ respectively. Both sellers have a reserve price of zero, namely $r_1 = r_2 = 0$.

\begin{figure}[h]
	\centering
	\includegraphics[scale=0.25]{../figures/illustrative_ex.jpg}
	\captionsetup{justification=centering, textfont=small, format=hang}
	\caption{Assignment market with two buyers and two sellers.}
	\label{fig:example}
\end{figure}

Two possible outcomes exist: In the first one, indicated with a grey solid line in Figure \ref{fig:example}, buyer $B_1$ is assigned to seller $S_1$ and $B_2$ to $S_2$, achieving an aggregate surplus of \$10. If seller $S_1$ charges \$1 and $S_2$ charges \$2.5, then the utility of $B_1$ is $4-1=3$ and that of buyer $B_2$ is $6-2.5=3.5$. $B_2$ therefore has higher utility by acquiring item of seller $S_2$ than that of $S_1$ for the given price vector ($3-1 = 2 < 3.5$). In contrast, $B_1$ would prefer to acquire the item of $S_2$, but achieves a higher utility (in terms of price) by trading with $B_2$.  

In the second outcome,  indicated by a dotted line in Figure \ref{fig:example}, buyer $B_1$ is assigned to seller $S_2$ and $B_2$ to $S_1$ with an aggregate surplus of \$13. However, there is no price vector that makes this outcome in the core. Suppose prices for the items of sellers $S_1, S_2$ were set to 3 and 2 respectively. At this price, $B_1$ can no longer afford the item of $S_1$ but achieves a payoff of $10-2=8$ from buying the item of $S_2$. Buyer $B_2$ has a payoff of $3-3=0$ from buying the item of $S_1$, and profits from switching to $S_2$, with corresponding payoff $6-2=4$. $S_2$ cannot charge more than \$2 because this would otherwise exceed the budget of buyer $B_1$. As a result $B_2$ and $S_2$ always strictly prefer being assigned to one another. Thus, the surplus-maximizing outcome is not in the core, and the auctioneer must account for budgets in the process of determining such an outcome.

A question that arises from the budget constraints in this example regards that of Pareto optimality of a surplus-maximizing solution. Without budget constraints but with quasilinear utility functions, aggregate surplus maximization is a necessary condition for Pareto efficiency (see \citet[10.D]{mas1995microeconomic}). 
Budget constraints restrict bidders from Pareto improvements, thus leading to outcomes that are inefficient if the auctioneer did not consider budget constraints. In markets with soft budget constraints there might be opportunities for Pareto improvements through side payments or other arrangements. For instance, if a buyer could borrow money or negotiate a payment plan, she might be able to acquire an item, reflecting her higher valuation and potentially leading to a Pareto improvement. We look at markets with hard budget constraints, where this is not possible. Such a constraint might be a spending limit imposed by a third party for a particular transaction. Ignoring this limit would lead to infeasible outcomes similar to ignoring allocation constraints that an auctioneer needs to consider. As long as the utilities of agents are assumed to be comparable in a quasilinear utility function, an auctioneer might still want to maximize surplus as in the case without this hard spending limit. However, in order not achieve a stable outcome, she would aim to find a surplus-maximizing outcome among those in the core.
} 


\section{Related Literature}\label{sec:related_literature}

Two-sided matching markets describe markets where buyers want to win at most one item (also known as the unit-demand model) and sellers sell only one item. Buyers and sellers are disjoint sets of agents and each buyer forms exclusive relationships with a seller. Such markets are central to economics, and there is substantial literature. For instance, the well-known marriage model by \citet{gale1962college} assumes ordinal preferences and non-transferable utility. \citet{koopmans1957assignment} analyzed such markets with quasilinear utility functions (aka. assignment markets or assignment model), and \citet{shapley1971assignment} showed that the core of this game is nonempty and encompasses all competitive equilibria. 
While their model assumed access to all valuations, \citet{demange1986multi} showed that an ascending auction with only demand queries results in a competitive equilibrium at the lowest possible price, i.e., at the competitive equilibrium price vector that is optimal for buyers. In such an auction, the auctioneer specifies a price vector (the demand query) in each round, and buyers respond with their demand set, i.e., the set of goods that maximize payoff at the prices. 

The housing market by \citet{shapley1974cores} is an example of a market without transferable utility or monetary funds. 
In the housing market, each agent is endowed with a good or house, and each agent is interested in one house only. 
The goal of this market is to redistribute ownership of the houses in accordance with the ordinal preferences of the agents. 
In such housing markets, the core is nonempty. If no agent is indifferent between any two houses, then the economy has a unique competitive allocation, which is also the unique strong core allocation. An allocation belongs to the strong core, if no coalition of buyers and sellers can make all members as well off and at least one member better off by trading items among themselves. When referring to the notion of core, we assume an allocation belongs to the core if no coalition can lead to all members' utilities improved when redistributing items amongst themselves. 

Related to our paper are house allocation problems without endowments, which describe markets without transferable utility but unit-demand agents. \citet{svensson1999strategy} showed that a mechanism for this problem is strategy-proof, non-bossy, and neutral if and only if it is a serial dictatorship, which is strongly Pareto-optimal, but not envy-free. \citet{bogomolnaia2001new} show that the random serial dictatorship mechanism is the unique strategy-proof, ex-post efficient mechanism that eliminates strict envy between agents with the same preferences. Note that we rule out buyers with zero budgets in our model such that utility can be transferred up to the budget constraint.

There is literature on non-quasilinear models where small changes of prices do not lead to a discontinuous change of the buyers' utilities as is the case with hard budget constraints. \citet{alaei2016competitive} present a characterization of utilities in competitive equilibria of two-sided matching markets in which the utility of each agent depends on the choice of the partner and the terms of the partnership. 
In contrast, \citet{quinzii1984core} analyze an exchange model with multiple agents with unit demand and transferable but non-quasilinear utility. 
Buyers derive utility from at most one good and a transfer of money. Sellers aim at obtaining the highest possible price above a reservation level. 
They show that outcomes are in the core and they are competitive equilibria. In a related model, \citet{gale1984equilibrium} show that a competitive equilibrium always exists. 
These models allow for "soft" budget constraints with finite financing costs. 
Hard budget constraints as in our paper are different, because buyers must not spend more than a certain amount of money, and the constraints cannot be relaxed. 

\citet{zhou2017multiitem} provide a multi-item auction for computing core-allocations. In contrast to their algorithm, ours only relies on demand queries, so buyers are not required to report their budgets. 
\citet{zhou2020serial} analyze an assignment market with imperfect transferability of utility and income effects. Their serial Vickrey mechanism finds a minimum price equilibrium, but their utility model is different and allows for income effects. \citet{herings2022competitive} introduce the notion of a quantity-constrained competitive equilibrium and establish existence and equivalence to core outcomes for a more general utility model. Recently, \citet{HERINGS2024105799} generalized previous works by allowing for more complex contract sets and preference structures. Their notion of expectational equilibrium provides a more generalized and flexible equilibrium concept that connects and extends existing theories in economic matching models.

Closer to our paper is the work by \citet{che2013assigning} who study different methods of assigning a good to budget-constrained agents. They show that mechanisms that assign the good randomly and allow resale may outperform the competitive market.
We focus on deterministic mechanisms like \cite{laan2016ascending}, who propose an ascending auction for the assignment market that results in a \emph{equilibrium under allotment}. This is in general not an outcome in the core. \citet{talman2015efficient}  analyze the same model and introduce an elegant algorithm that finds a core solution in an assignment market with hard budget constraints, \MB{but there are no incentive compatibility guarantees.} They find an outcome in the core, i.e., one that is weakly Pareto-optimal. 


We show that without an additional assumption, incentive compatibility and the core are not always possible with hard budget constraints. Importantly, even with access to all valuations and budgets, we can show that finding a strongly Pareto-optimal outcome in the core is NP-hard for larger problem sizes and the aggregate surplus of simple iterative auctions can be arbitrarily low with possibilities for Pareto improvements. It was shown that for combinatorial exchanges, finding an outcome that maximizes aggregate surplus in the core is in even higher complexity classes of the polynomial hierarchy \citep{bichler2021or}. Overall, computational complexity is a key concern in this model and a barrier for practical implementations. The paper by \citep{bichler2021or} does not discuss incentive compatibility, because \citet{Dobzinski2008} already showed that incentive-compatible and Pareto-efficient auctions are impossible for multi-unit auctions with financially constrained buyers. The impossibility result by \citet{Dobzinski2008} requires multi-unit demand, and does not extend to the assignment model with unit-demand buyers, which is the focus of this paper.

\section{Preliminaries} \label{sec:preliminaries}

An assignment market \MB{$M=(\mcal{B}, \mcal{S}, V, B, R)$} 
consists of two disjoint sets of agents $\mcal{B}$ and $\mcal{S}$, representing buyers, i.e., bidders, $i \in \mcal{B} = \{1,2,\dots, n\}$ and goods $j \in \mcal{S} = \{1,2,\dots, m\} \cup \{0\}$.  We identify good $j$ with the seller owning it, i.e., each seller owns one good. The $0$-item is a dummy item and does not have value to any bidder, meaning that a bidder receiving good $0$ corresponds to no real good. Additionally, agents' preferences are given by each bidder $i$'s valuation $v_i: \mcal{S} \rightarrow V \subseteq \Z_{\geq 0}$ with $v_i(0)=0$ and budget $b^i \in B \subseteq \Z_{> 0}$, as well as each seller $j$'s reserve values/ask price \MB{$r_j : \mcal{B} \rightarrow R \subseteq \Z_{\geq 0}$}.
\footnote{The integrality assumption for valuations $v_i$ is made throughout in the literature starting with \citet{demange1986multi} or \citet{mishra2007ascending}. } The budget constraint is a limit on how much an agent is allowed to spend in a particular market, and it cannot be relaxed. 
A \emph{price vector} is a vector $p \in \R^{\mcal{S}}$ with $p(0) = 0$, assigning price $p(j)$ to every good. Bidders have quasilinear utilities, so if bidder $i$ receives item $j$ under prices $p$, their utility is $\pi_i(j,p) = v_i(j)-p(j)$, if $p(j) \leq b^i$, and $\pi_i(j,p) = -\infty$, otherwise. \MB{Such utility functions do not allow for affine transformations. }

An \emph{assignment} is represented as a map $\mu: \mcal{B} \rightarrow \mcal{S}$ from bidders to the items they receive, with $\mu \in X$ and $X$ describes the set of feasible assignments. Here, $|\mu^{-1}(\{j\})| \leq 1$ for all $j \neq 0$, so only the dummy good may be assigned to more than one bidder. 
\MB{
An \textit{outcome} of the market mechanism is a pair \((\mu,t)\) where 
\(\mu\colon \mathcal{B} \to \mathcal{S}\) 
is an assignment mapping each buyer \(i\) either to a seller \(\mu(i)\in \mathcal{S}\) or to 
the 0-item, which implies that the bidder is not matched to any seller, 
and \(t=(t_i)_{i\in \mathcal{B}}\in\mathbb{R}^\mathcal{B}\) is a transfer vector. We then \textit{induce} a price vector \(p\in\mathbb{R}^\mathcal{S}\) by \(p(j)\;=\;t_i
\text{ whenever } \mu(i)=j. \)
We will also write $(\mu,p)$ to describe an outcome \((\mu,t)\) with induced prices \(p\). 
}
\MB{In settings with budget constraints, we focus on outcomes where the induced price vector is such that no budget constraint is violated, i.e., $p(\mu(i)) \leq b^i$ for all $i \in \mcal{B}$ and only sold items may have a positive price: $p(j) > 0$ implies that $|\mu^{-1}(\{j\})|=1$.}

We also require that the price for the dummy item $0$ is $0$. For the sake of simplicity of the exposition, we assume all reserve prices to be equal to $0$ in what follows, as the results can be easily generalized by starting the auction at the reserve prices, instead of $0$. 
Given prices, the \emph{demand set} of bidder $i$ consists of the most preferred items $j$:
\[
D_i(p) = \left\{j\,:\, p(j) \leq b^i \, \wedge \, \pi_i(j,p) \geq \pi_i(k,p) \, \forall k \text{ with } p(k) \leq b^i\right\}.
\]

A deterministic \textit{market mechanism} $\mathcal{M}=(V,B,f,p)$ is defined by a set of reported valuations $V$ and budget constraints $B$ available to the bidders $i \in \mcal{B}$, an allocation function $f: V^n \times B^n \rightarrow X$, and a payment rule \MB{$p: X \rightarrow \R^n$}
for each agent $i$. We want such market mechanisms to be incentive-compatible.

\begin{definition}[Dominant-Strategy Incentive Compatibility (DSIC)]
Let \(\mathcal{M}~=~(V, B, f, p)\) be a mechanism, where each agent \(i\) has true value \(v_i\in V\) and true budget \(c_i\in B\).  For any assignment \(x\) and payment \(p\), define
\[
u_i(x,p \;;\,v_i,c_i)
\;=\;
\begin{cases}
v_i(x)\;-\;p, & p \le c_i,\\
-\infty,      & p > c_i.
\end{cases}
\]
Then \(\mathcal{M}\) is \emph{dominant-strategy incentive-compatible with budgets} if, for every agent \(i\), every true profile \((v,c)\in V^n\times B^n\), every misreport \((w_i,d_i)\in V\times B\), and every profile of other reports \((w_{-i},d_{-i})\in V^{n-1}\times B^{n-1}\) including any demand set $D_i(p)$ derived in an ascending auction, we have
\[
u_i\bigl(f((v_i, w_{-i}),(c_i,d_{-i}))_i,\;p((v_i, w_{-i}),(c_i,d_{-i}))_i \;;\,v_i,c_i\bigr) 
\;\ge\;
\]
\[
u_i\bigl(f\bigl((w_i,w_{-i}),(d_i,d_{-i})\bigr)_i,\;
      p\bigl((w_i,w_{-i}),(d_i,d_{-i})\bigr)_i
      \;;\,v_i,c_i\bigr).
\]
In other words, no matter what the others report, agent \(i\) can never gain by lying about her own value or budget.
\end{definition}

In a direct mechanism, this would mean that a bidder reveals his true values for the items and his budget constraint. An indirect mechanism might use demand or value queries. In a demand query, participants are asked to specify their demand set at the prices. In a value query, bidders are typically asked to reveal their private valuation or willingness to pay for the item without reference to a specific price. 
In an indirect and incentive-compatible mechanism as the one described below with only demand queries, a bidder would reveal his demand set up to his budget limit and then exit the auction or switch to the remaining item with the highest payoff. Many real-world auctions follow such a process. \MB{Standard ascending auctions (like English or clock auctions without proxy bidders) cannot be dominant-strategy incentive-compatible, because truth-telling isn't always best against arbitrary behavior. As usual, we restrict to ex-post incentive compatibility in this case, and follow definitions as in \citet{mishra2007ascending}.}

\begin{definition}[Ex-Post Incentive Compatibility (EPIC)]
Let \(\mathcal{M} = (V, B, f, p)\) be a mechanism.  For any agent \(i\), true valuation \(v_i\), true budget \(c_i\), reported valuation \(w_i\), and reported budget \(d_i\), define her utility under assignment \(x\) and payment \(p\) by
\[
u_i(x,p \,;\, v_i, c_i)
\;=\;
\begin{cases}
v_i(x) - p, & p \,\le\, c_i,\\
-\infty,    & p \,>\, c_i.
\end{cases}
\]
Then \(\mathcal{M}\) is \emph{ex-post incentive compatible with budgets} if, for every agent \(i\), every true profile \((v,c)\in V^n\times B^n\), and every misreport \((w_i,d_i)\in V\times B\), it holds that
\[
u_i\!\bigl(f(v,c)_i,\;p(v,c)_i\,;\,v_i,c_i\bigr)
\;\ge\;
u_i\!\bigl(f((w_i,v_{-i}),(d_i,c_{-i}))_i,\;
      p((w_i,v_{-i}),(d_i,c_{-i}))_i
      \;;\,v_i,c_i\bigr).
\]
In other words, even if \(i\) lies about her valuation \emph{and/or} her budget, she can never improve her (budget-constrained) payoff ex post.
\end{definition}

A design desideratum for the {outcome} $(\mu,p)$ of a mechanism is that of Pareto optimality (PO), aka. Pareto efficiency. (Weak) Pareto optimality means that an allocation is considered optimal if there is no other allocation that makes everyone strictly better off. In contrast, strong Pareto optimality requires that no reallocation can make some people better off without making others worse off, even if the others are indifferent.

\begin{definition}[Pareto optimality]
An outcome \((\mu, p)\) is (weakly) Pareto-optimal if there is no other allocation \((\mu', p')\) such that:
\begin{align*}
v_i(\mu'(i)) - p'(\mu'(i)) &> v_i(\mu(i)) - p(\mu(i)) \quad \text{for all } i \in \mathcal B, \\
p'(j) &> p(j) \quad \text{for all } j \in \mathcal  S
\\
p'(\mu'(i)) &\leq b^i \quad \forall i  \in \mathcal B.
\end{align*} 
\end{definition}

\begin{definition}[Strong Pareto optimality]
An outcome \((\mu, p)\) is strongly Pareto-optimal if there is no other allocation \((\mu', p')\) such that:
\begin{align*}
v_i(\mu'(i)) - p'(\mu'(i)) &\geq v_i(\mu(i)) - p(\mu(i)) \quad \text{for all } i \in \mathcal B, \\
p'(j) &\geq p(j) \quad \text{for all } j \in \mathcal  S
\\
p'(\mu'(i)) &\leq b^i \quad \forall i  \in \mathcal B
\\
\text{and for some $i \in \mathcal B$: } v_i(\mu'(i)) - p'({\mu'(i)}) &> v_i(\mu(i)) - p(\mu(i)),
\\
\textbf{or } \text{for some $j \in \mathcal S$: } p'(j) &> p(j).
\end{align*} 
\end{definition} 
We are also interested in an outcome in the core of the market game: a pair $(\mu,p)$ is a core outcome if there does not exist any pair $(\mu', p')$ for some coalition such that the utility for each member of the coalition is strictly higher. A strong core allocation means that no coalition can make all its members at least as well off, with at least one member strictly better off \citep{talman2015efficient}.


\hide{
\begin{definition}[Core]
An outcome \((\mu, p)\), where \( \mu: \mcal{B} \rightarrow \mcal{S} \) is the assignment of goods to bidders and \( p: \mcal{B} \rightarrow \Z_{\geq 0} \) is the price paid by each bidder, is said to be in the \textit{(weak) core} if 
there does not exists any coalition $C = B \cup S \subseteq \mathcal B \cup \mathcal S$ and outcome $(\mu',p')$, such that $\mu'(B) = S$ and
\begin{align*}
v_i(\mu'(i)) - p'(\mu'(i)) &\geq v_i(\mu(i)) - p(\mu(i)) \quad \text{for all } i \in B, \\
p'(\mu'(i)) &\geq p(\mu'(i)) \quad \text{for all } \mu'(i) \in S
\\
p'(\mu'(i)) &\leq b^i \quad \forall i  \in B
\\
\text{For some $i \in B$: } v_i(\mu'(i)) - p'({\mu'(i)}) &> v_i(\mu(i)) - p(\mu(i)),
\\
\text{AND for $\mu'(i) \in S$: } p'(\mu'(i)) &> p(\mu'(i)) 
\end{align*}
The outcome \((\mu, p)\) is in the \textit{strong core}, if there does not exist any coalition $C = B \cup S \subseteq \mathcal B \cup \mathcal S$ and outcome $(\mu',p')$ such that
\begin{align*}
\text{For $i \in B$: } v_i(\mu'(i)) - p'(\mu'(i)) &> v_i(\mu(i)) - p(\mu(i)),
\\
\text{OR for $\mu'(i) \in S$: } p'(\mu'(i)) &> p(\mu'(i)) 
\end{align*}
\end{definition}
}


\begin{definition}[Core and Strong Core]
An outcome $(\mu,p)$ is in the core of the assignment market, if and only if there does not exist a coalition of a single buyer $i$ and single seller $j$ such that
$v_i(j) - p(j) > v_i(\mu(i)) - p(\mu(i))$ \textbf{and} $p(j) < b^i.$ 
An outcome $(\mu,p)$ is in the strong core, if and only if there does not exist a coalition of a single buyer $i$ and single seller $j$ such that
$v_i(j) - p(j) > v_i(\mu(i)) - p(\mu(i))$ \textbf{or} $p(j) < b^i$.
\end{definition}

The idea of a blocking pair $(i,j)$ when the outcome is not in the core is that both bidder $i$ and seller $j$ would strictly increase their utility, if $i$ received item $j$ instead of $\mu(i)$: if $i$ pays $p(j)+\varepsilon$ for item $j$, then still $\pi_i(j,p)-\varepsilon > \pi_i(\mu(i),p)$, and at the same time, the profit of seller $j$ is increased by $\varepsilon$.

Note that in our setting, the definition of the core provided above is equivalent to the standard definition in the literature: an outcome $(\mu,p)$ is said to be in the core if there are no subsets $\mcal{B}' \subseteq \mcal{B}$ and $\mcal{S}' \subseteq \mcal{S}$ and an outcome $(\mu', p')$ on $\mcal{B}' \times \mcal{S}'$ such that $\pi_i(\mu'(i),p') > \pi_i(\mu(i),p)$ for all $i \in \mcal{B}'$ and $p'(j) > p(j)$ for all $j \in \mcal{S}'$ (see for example \citet{zhou2017multiitem}). These definitions can easily be seen to be equivalent: first suppose that such subsets $\mcal{B}'$ and $\mcal{S}'$ do exist. Then both sets are nonempty. In particular, let $i \in \mcal{B}'$ and $j = \mu'(i)$. Then $p'(j) > p(j)$, so $p(j) < b^i$. Furthermore, we have $\pi_i(j,p) > \pi_i(j,p') > \pi_i(\mu(i),p)$, so $(i,j)$ is a blocking pair. On the other hand, if $(i,j)$ is a blocking pair, then as in the above paragraph, we can set $p'(j) = p(j)+\varepsilon$ and get $\pi_i(j,p') > \pi_i(\mu(i),p)$ and $p'(j) > p(j)$. Thus we can choose $\mcal{B}' = \{i\}$, $\mcal{S}' = \{j\}$, $\mu'(i) = j$ and $p'(j)=p(j)+\varepsilon$ in the alternative definition.

\section{Assignment markets without binding budget constraints}

Let us first assume an assignment market where budget constraints are not binding and can be ignored. This is the model analyzed by \citet{shapley1971assignment}, who compute a competitive equilibrium. 

\begin{definition}[Competitive equilibrium]\label{def:ce}
	An outcome $(\mu,p)$ is a \emph{competitive equilibrium}, if $\mu(i) \in D_i(p)$, i.e., it is an element of their demand set for all bidders $i$.
\end{definition}

\hide{
\begin{definition}[Competitive equilibrium]
A set of prices \(\{p\}\) and an allocation \(\{\mu\}\) form a competitive equilibrium if:
\begin{itemize}
    \item Each buyer \(i\) maximizes their utility given the prices $p$, i.e., \(\mu(i)\) maximizes \(v_i(\mu(i)) - p(\mu(i))\).
    \item Each seller \(j\) maximizes their payoff $p(j)$.
    \item Markets clear: the total demand for each good equals the total supply.
\end{itemize}
\end{definition}
}

The model of \citet{shapley1971assignment} maximizes (aggregate) surplus, i.e., the sum of values $v_i$ of the market participants. 
\[
\max \left\{\sum_{i=1}^n v_i(\mu(i)) \,:\, \mu \text{ is an  assignment}\right\}
\]

The dual variables ($p(j)$) of the following linear program (LP) constitute competitive equilibrium prices, which follows from complementary slackness conditions in linear programming.

\begin{align}\label{eqn:assignment_lp}
	\max &\, \sum_{i=1}^n \sum_{j=1}^m x_{ij}v_i(j) \\
	\text{s.t.} &\, \sum_{i=1}^n x_{ij} \leq 1 \, \forall j=1,\dots,m &\, (p(j)) \nonumber \\
	&\, \sum_{j=1}^m x_{ij} \leq 1 \, \forall i=1,\dots,n &\,(\pi_i) \nonumber \\
	&\, x \geq 0 \nonumber
\end{align}

The objective of this program maximizes aggregate surplus where we assume that sellers have no value for the items in this case. The variables in parentheses denote the corresponding duals with $\pi_i$ as payoff of the buyers. This assignment problem is well-known to have an integral optimal solution and can be solved in $O(n^3)$ \citep{kuhn1955hungarian, edmonds1972theoretical}. An integral solution $x$ corresponds to an assignment $\mu$ where $x_{ij} = 1$ is equivalent to $\mu(i) = j$.  

\citet{shapley1971assignment} show that in this model the core is not empty and that the core corresponds to the set of competitive equilibrium outcomes.
A folk result is that in markets with quasilinear utilities one need not distinguish between core and strong core.

\begin{proposition}
In a market where agents have quasilinear utilities and no budget constraint binds, the core coincides with the strong core.
\end{proposition}

\proof{Proof: }
Let $(\mu,p)$ be a core allocation. Suppose it is not in the strong core.
Then there exists $(\mu',p')$ and a coalition satisfying the equations above. In particular, there is one agent that strictly improves their utility, say, by $\varepsilon$. But by transferable utility, this agent can distribute a small portion $\varepsilon$ of his gained utility evenly among all agents in the coalition, making everyone better off.
\qed \endproof


The next proposition summarizes equivalences of different notions of markets with quasilinear utilities when budget constraints do not bind.
\begin{proposition}(\citet{bikhchandani1997competitive}). \label{prop:unbinding_equivalences}
	Suppose that in an assignment market \\$b^i > v_i(j)$ for all $i$ and $j$, and let $(\mu,p)$ be an outcome. Then the following statements are equivalent:
	\begin{enumerate}
		\item $(\mu,p)$ is a strong core outcome.
		\item $(\mu,p)$ is a competitive equilibrium.
		\item If $x_{ij} = 1 \iff \mu(i) = j$ then $x$ is an optimal solution of the linear program (\ref{eqn:assignment_lp}), and $p(j)$ is dual optimal.
		\item $(\mu,p)$ is a surplus-maximizing assignment.
	\end{enumerate}
\end{proposition}

If an allocation is in the strong core, then it is impossible for any coalition to find a reallocation that makes at least one member strictly better off without making others worse off. 
This directly satisfies the condition for strong Pareto optimality, which requires that no such Pareto-improving reallocation is possible in the grand coalition (the entire set of agents). If no coalition, including the grand coalition, can improve upon the allocation without making someone worse off, then the allocation must be strongly Pareto-optimal. 
Overall, it is well-known that in a market with quasilinear bidders and no binding budget constraints a surplus-maximizing outcome is necessary for strong Pareto optimality \citep{mas1995microeconomic, negishi1960welfare}. 
An optimal solution to the linear program \ref{eqn:assignment_lp} corresponds to a surplus-maximizing assignment, and the dual variables $p(j)$ constitute competitive equilibrium prices where no participant would want to deviate and supply equals demand. 


Other objective functions do not have the same properties. For example, prices derived from the dual of a weighted aggregate surplus function reflect the weighted average valuations, which do not align with buyers' actual marginal willingness to pay. Introducing weights or alternative objective functions leads to allocations where aggregate surplus is not maximized, and Pareto improvements are possible under quasilinearity by redistributing money. In contrast, maximizing (unweighted) aggregate surplus leads to dual variables that reflect marginal value and lead to a competitive equilibrium where no agent's payoff can be improved without making another one worse off. A competitive equilibrium is in the strong core and therefore it is strongly Pareto-optimal. This describes a version of the First Welfare Theorem for assignment markets \citep{Arrow1954}.

\section{Assignment markets with binding budget constraints}

In this section, we discuss the core and Pareto optimality in assignment markets with binding budget constraints.

\subsection{The Core}

With hard budget constraints, maximizing aggregate surplus does not lead to an outcome in the core in our model.
Suppose we have a single good, a bidder 1 with $v_1=6$ and a budget of $b^1=3$ and a second bidder with a value of $v_2=10$ and a budget of $b^2=2$. Assigning the good to bidder 2 is not in the core as the seller can form a blocking coalition with bidder 1. In what follows, we aim for outcomes in the core, which are strongly Pareto-optimal in this environment. 



As we discussed in the previous section, without budget constraints the set of competitive equilibria coincides with the (strong) core, which in particular implies strong Pareto optimality. With budget constraints this statement is not true anymore. It is easy to see that in a market with budget constraints competitive equilibria might fail to exist, as the following example shows: 
\begin{example}\label{exp1}
	Consider two bidders $1, 2$ and one item $A$. Suppose that $v_1(A) = 6$ and $v_2(A) = 100$. Both bidders have the same budget $b^1 = b^2 = 1$. It is easy to see that there are two core outcomes: either bidder $1$ or bidder $2$ receives $A$ for a price of \$1, while the other bidder does not receive an item. However, the item is still in his demand set and thus not a competitive equilibrium.
\end{example}

This point was already raised by \citet{talman2015efficient}, who suggested focusing on a core outcome in assignment markets with binding budget constraints. Core outcomes are stable and thus desirable. They also show that in an assignment market \textit{with} budget constraints the core can differ from the strong core and the strong core can be empty.

\begin{example}
Similar to the previous example, there are again two buyers and a single good; let both have value $v_i(A)=6$ for the item and a budget of $b^i=1$. The (weak) core allocations are those where one buyer gets the item and pays \$1. These are no \emph{strong} core allocations, since the buyer who does not receive the item can form a coalition with the seller. In this coalition, the seller's utility would not change (he simply receives $1$ from the buyer), so it \emph{weakly} increases, while the utility of the buyer of the new coalition \emph{strictly} increases (from $0$ to $6-2 = 4$). 
\end{example}

While strong core solutions are demanding and might not exist, it is desirable to look for strongly Pareto-optimal outcomes in markets with budget constraints. In a quasilinear utility model of an assignment market, $v_i(j)$ represents the value or willingness-to-pay of good $j$ in terms of money $p$. 
As in the case without binding budget constraints, one might care about assigning a good to a bidder whose value is a billion as compared to one whose value is a million dollars, even though it only gets a payment of a million dollars no matter to whom it assigns the good to. 


Example \ref{ex:example1} shows that the allocation maximizing surplus among assignments in the core are strongly Pareto-optimal. 
One might argue that this is a knife-edge case and that if bidder 1 had a budget of $1+\varepsilon$, then the only core outcome would be to assign the good to bidder 1. However, budgets do not need to be equal as the following example shows.

\begin{example}\label{ex:strongPO_weak}
Consider three buyers 1, 2, 3, and two items A and B, with the following values and budget constraints in Table \ref{ex:strongPO_table}. At a price of 1, the auctioneer could assign A to bidder 1 and B to bidder 2 and exclude bidder 3. This is a strongly Pareto-optimal outcome, because one cannot find another feasible allocation at these prices in which at least one individual is strictly better off without making anyone else worse off. Changing the price would lead either the buyer or the seller to having lower utility. 
However, the assignment of $B$ to bidder 3 is a weak blocking pair. If the seller of $B$ were to sell to bidder $3$, he would not decrease his utility, but this constitutes a core outcome. A strong core solution would be, to assign item $A$ to bidder 1 and $B$ to bidder 2, both for a price of 2.

\begin{table}[ht!]
  \centering
  \caption{Valuations and Budgets of Bidders}\label{ex:strongPO_table}
  \label{tab:valuations}
  \begin{tabular}{c c c c}
    \toprule
    Bidder & $v_i(A)$ & $v_i(B)$ & $b^i$ \\
    \midrule
    1 & 100 & 0   & 100 \\
    2 & 100 & 100 & 50  \\
    3 & 0   & 100 & 1   \\
    \bottomrule
  \end{tabular}
\end{table}


\end{example}


This example illustrates that the strong core is a sufficient but not a necessary condition for strong Pareto optimality with binding budget constraints. Note that at a specific price vector $p$ there might be core outcomes, while at another price vector $q$ there can be a strong core outcome $(\nu,q)$. In this example, we have a strongly Pareto-optimal outcome in the core that is not in the strong core.
Similarly, an outcome in the core is weakly Pareto-optimal, but not vice versa. \citet{talman2015efficient} show that at least one outcome in the core always exists in our model with binding budget constraints, but that the strong core can be empty.

\subsection{Surplus Maximization and Strong Pareto Optimality}

\MB{Example \ref{ex:example1} suggests that the surplus-maximizing allocation among the ones in the core can guarantee strong Pareto optimality, i.e., there is no re‐assignment at the given prices that Pareto‐improves at least one agent. Based on this, we define a surplus-maximizing outcome in the core as one that maximizes the sum of valuations $\sum_{i=1}^n v_i(\mu(i))$ such that $(\mu,p)$ is a feasible core outcome.}



\begin{proposition}\label{prop:strongPO}
In an assignment market with quasilinear utility functions and hard budget constraints, selecting a core allocation $\mu^*$ at prices $p$ that maximizes aggregate surplus among all core allocations at $p$ is a sufficient condition for the outcome $(\mu^*, p)$ to be strongly Pareto-optimal.
\end{proposition}

\proof{Proof:}
First, observe that if an allocation $\mu$ is in the core with a certain price vector $p$, then no seller can be made better off at this price vector by assigning a different buyer. This is because a seller's utility depends solely on the price $p(j)$ of their good $j$, not on which buyer purchases it. \MB{Indeed, in our model each seller $j$’s payoff is exactly $p(j)$, independent of which bidder wins $j$, so any reassignment at price $p(j)$ leaves all sellers exactly as well off.} Since the price vector $p$ is fixed in the core outcome, reassigning the good to a different buyer at the same price cannot increase the seller's utility.

We claim that selecting a core allocation $\mu^*$ that maximizes aggregate surplus among all core outcomes at prices $p$ is a strongly Pareto-optimal outcome. Now, suppose for contradiction that the outcome $(\mu^*, p)$ is \emph{not} strongly Pareto-optimal. Then there exists another feasible allocation $\mu'$ at the same prices $p$ such that:
\[
v_i(\mu'(i)) - p(\mu'(i)) \geq v_i(\mu^*(i)) - p(\mu^*(i)) \quad \text{for all } i \in \mathcal{B},
\]

with strict inequality for at least one buyer $i$. \MB{We only compare to reallocations $\mu'$ that are \emph{feasible} at prices $p$, i.e.\ whenever $\mu'(i)=j$ we require $p(j)\le b^i$ so that each bidder can afford her assignment.
Now, we claim that \((\mu',p)\) is also in the core at prices \(p\).  Indeed:}

\noindent 1.  \textit{Budget-feasibility.}  By assumption \(\mu'\) is feasible, so whenever \(\mu'(i)=j\), we have \(p(j)\le b^i\).  \\
2.  \textit{No buyer-seller blocking.}  Since \((\mu^*,p)\) is in the core, for \textit{every} buyer \(i\) and every item \(j\),
   \[
     \pi_i^* + p(j)\;\ge\;v_i(j).
   \]
   But here
   \[
     \pi_i(\mu'(i),p)
     = v_i(\mu'(i)) - p(\mu'(i))
     \;\ge\;\pi_i^*
     \quad\Longrightarrow\quad
     \pi_i(\mu'(i),p) + p(j)
     \;\ge\;
     \pi_i^* + p(j)
     \;\ge\;
     v_i(j).
   \]
   Therefore, no buyer-seller pair can block \((\mu',p)\).  \\
3.  \textit{No seller blocking.}  As before, every seller’s payoff is exactly \(p(j)\), so reassigning goods at the same price vector cannot make any seller better off.
\MB{Hence \((\mu',p)\) satisfies the core-stability conditions at \(p\).  
}

\noindent The total surplus under allocation $\mu'$ is:
\[
S(\mu') = \sum_{i \in \mathcal{B}} v_i(\mu'(i)) = \sum_{i \in \mathcal{B}} \left[ v_i(\mu'(i)) - p(\mu'(i)) \right] + \sum_{j \in \mathcal{S}} p(j).
\]
Similarly, the total surplus under $\mu^*$ is:
\[
S(\mu^*) = \sum_{i \in \mathcal{B}} v_i(\mu^*(i)) = \sum_{i \in \mathcal{B}} \left[ v_i(\mu^*(i)) - p(\mu^*(i)) \right] + \sum_{j \in \mathcal{S}} p(j).
\]

\noindent Subtracting the two expressions, we obtain:
\[
S(\mu') - S(\mu^*) = \sum_{i \in \mathcal{B}} \left[ \left( v_i(\mu'(i)) - p(\mu'(i)) \right) - \left( v_i(\mu^*(i)) - p(\mu^*(i)) \right) \right] = \sum_{i \in \mathcal{B}} \left[ \pi_i(\mu'(i), p) - \pi_i(\mu^*(i), p) \right].
\]

\noindent Since $\pi_i(\mu'(i), p) \geq \pi_i(\mu^*(i), p)$ for all $i$, with strict inequality for at least one $i$, it follows that:
\[
S(\mu') - S(\mu^*) > 0.
\]
This contradicts the assumption that $\mu^*$ is a core allocation that maximizes aggregate surplus among all core allocations.
\qed \endproof


While a core outcome that is maximizing aggregate surplus is a sufficient condition for strong Pareto optimality, it is no sufficient condition for the strong core, which is a much stronger condition. This was already shown in example \ref{ex:strongPO_weak}. The following example \ref{exp:strongcore} is similar, but also shows that a strong core outcome (if it even exists) can have much lower aggregate surplus than one that is strongly Pareto-optimal. 

\begin{example}\label{exp:strongcore}
Consider three bidders $1$, $2$, $3$ and two goods $A$ and $B$ with values and budgets as in Table \ref{ex:strong}. 
Then we have the \textit{weak} core solution $p(A) = 0$, $p(B)=1$ where $1$ gets good $A$ and $2$ gets $B$ with aggregate surplus of $1002$. However, there is a strong core solution with $p(A) = 1$, $p(B) = 2$, where $1$ gets $B$ and $2$ gets $A$ with aggregate surplus of $5$.  In the core solution the third bidder could improve his utility, while this is not possible in the strong core solution.
\begin{table}[h]
\centering
\caption{Valuations and Budgets for Bidders}\label{ex:strong}
\begin{tabular}{c|c c c}
\hline
{Bidder} & ${v(A)}$ & ${v(B)}$ & ${b^i}$ \\ 
\hline
1 & $2$     & $3$     & $10$ \\
2 & $2$     & $1000$  & $1$      \\
3 & $0$     & $2$     & $1$      \\
\hline
\end{tabular}
\end{table}
\end{example}

Next, we show that similar to the standard quasilinear case without binding budget constraints, bidders in a core outcome either maximize their utility at core prices or they cannot afford them. The strong core (if it exists) is even equivalent to a competitive equilibrium in such assignment markets with binding budget constraints.

\begin{proposition}
In a core outcome of an assignment market with quasilinear bidders and budget constraints $b^i$, bidder $i$ gets an item assigned for which his payoff is at least
\[
\max_{j: p(j) < b^i} v_i(j)-p(j).
\]
\end{proposition}

\proof{Proof:}
Suppose bidder $i$'s payoff is smaller than
\[
\max_{j: p(j) < b^i} v_i(j)-p(j)
\] 
and let $j$ be a maximizer of this expression. Then, $i$ could pay $p(j)+\varepsilon$ to seller $j$ in order to receive $j$, since $p(j) < b^i$. This would make both better off.
\qed \endproof
\begin{proposition}
In a strong core outcome of an assignment market with quasilinear bidders and budget constraints $b^i$, bidder $i$ gets an item assigned for which his payoff is equal to
\[
\max_{j: p(j) \leq b^i} v_i(j)-p(j).
\]
\end{proposition}
\proof{Proof:}
Let $i$ be a bidder and $j$ be the item he receives. Suppose first that $j$ maximizes $i$'s utility subject to the budget constraint. Let $k$ be any other item. If $p(k) > b^i$, $(i,k)$ is no blocking pair, since the seller of $k$ would gain less money selling to $i$. On the other hand, if $p(k) \leq b^i$, $v_i(k)-p(k) \leq v_i(j)-p(j)$. To strictly improve the seller's utility, we would have to raise $p(k)$, which would strictly decrease the buyer's utility. To strictly improve the buyer's utility, we would have to decrease $p(k)$, decreasing the seller's utility. It follows that there is no blocking pair.

Now suppose that $j$ does not maximize $i$'s utility, i.e., there is $k$ with $p(k) \leq b^i$ and $v_i(k)-p(k) > v_i(j)-p(j)$. Then $i$ could buy $k$, not decreasing seller $k$'s utility, and strictly increasing $i$'s utility. Thus, $(i,k)$ is a weakly blocking pair, and the outcome is not in the strong core.
\qed \endproof

\begin{corollary}
The strong core is equal to the set of competitive equilibria.
\end{corollary}

The result follows from the definition of a demand set. Remember that without binding budget constraints, the optimization problem maximizing aggregate surplus allows us to derive a competitive equilibrium, which is a strong core solution and thus strongly Pareto-optimal in this case. We showed in Proposition \ref{prop:strongPO} that if we select a core outcome that maximizes aggregate surplus at the prices, then it is strongly Pareto-optimal. There can also be multiple outcomes in the strong core, which are all strongly Pareto-optimal. However, some might have higher aggregate surplus than other strong core solutions. By definition, there cannot be any of these strong core outcomes that allows to increase the utility of one or more participants without making any other agent worse off, even at different prices. 

Consider Example \ref{exp:strongcore} without bidder 3. Then we have the \textit{strongly} \MB{Pareto optimal} solution $p(A) = 0$, $p(B)=1$ where $1$ gets good $A$ and $2$ gets $B$ with a surplus of $1002$. However, there is also a strong core solution with $p(A) = 1$, $p(B) = 2$, where $1$ gets $B$ and $2$ gets $A$ with an aggregate surplus of $5$. 
Figure \ref{fig:concepts} summarizes the insights from this section.

\begin{figure}[htp!]
\centering
Weak PO $\Leftarrow$ Core $\Leftarrow$ Strong Core $\Leftrightarrow$ CE $\Rightarrow$ Strong PO \\
Core $\land$ Max. Surplus $\Rightarrow$ Strong PO
\caption{Relationship between versions of the Core, Pareto Optimality (PO), and Competitive Equilibrium (CE) }\label{fig:concepts}
\end{figure}

Note that in the standard quasilinear model without binding budget constraints the core, the strong core, weak and strong Pareto-optimality collapse.

\subsection{Incentive Compatibility and Core Stability}

It is well-known that in the model without binding budget constraints, an ex-post incentive-compatible ascending auction exists \citep{demange1986multi}, which generalizes the single-object clock auction.  
\MB{Our first theorem proves that, with hard budget constraints in the assignment model, incentive compatibility and the definition of the core can be incompatible. Incentive compatibility comprises DSIC and EPIC. }

\begin{theorem}\label{thm:ic}
In assignment markets with quasilinear preferences and budget-constrained bidders with unit demand, there is no incentive-compatible deterministic mechanism terminating in the core for every input. 
\end{theorem}

\proof{Proof:}
By the direct revelation principle \citep{gibbard1973manipulation}, we may assume that bidders report their exact valuations, as well as their budgets to the auctioneer.
Consider a market with three bidders $1,2,3$ and two items $A, B$. Let $\mcal{M}((v_1,b^1),\dots,(v_3,b^3))$ denote a mechanism that maps the bidders' reported valuations and budgets to the core with respect to their reports.

We consider instances of the above market where all bidders have the same values for both items: $v_i(A)=v_i(B)=10$ for $i=1,2,3$. Let us consider two instances, where the bidders vary their reported budget.

\begin{enumerate}
	\item If all bidders report $b^i=1$ for $i=1,2,3$, then obviously, since there are only two items, one bidder does not receive one: for $(\mu,p) = \mcal{M}((v_1,1),(v_2,1),(v_3,1))$, there is an $i$ such that with positive probability $\alpha > 0$, with $\mu(i) = 0$. Without loss of generality, we assume that $i=3$. It is easy to see that for core-stability to hold, bidders $1$ and $2$ both receive an item, and that we can set $p(A)=p(B)=1$. Bidders $1$ and $2$ have utility $9$, while bidder $3$ has utility $0$.
	\item If bidder $3$ reports $b^3 = 2$, and the other bidders report $b^1=b^2 = 1$, then clearly bidder $3$ receives an item in any core-stable outcome, and without loss of generality $\mu(3)=B$. Also, the other item $A$ must necessarily be assigned to some bidder. Again, without loss of generality, we assume that $\mu(1) = A$ and $\mu(2) = 0$. It is easy to see that $p(A)$ must be equal to $1$ in a core outcome. Additionally, we must have $p(A)=p(B)$, since otherwise bidder $3$ would strictly prefer item $A$ to item $B$, which would not be envy-free. Thus, $p(A)=p(B)=1$, and bidder $3$ has a utility of $9$.
\end{enumerate}
This already shows that $\mcal{M}$ is not incentive-compatible: If all bidders' true budgets are equal to $1$ and they report truthfully, bidder $3$ has a utility of $0$. However, if bidder $3$ misreports $b^3 = 2$, they would receive an item and have a utility of $9$. Note that $p(B)=p(A)=1$ in this case, so bidder $3$ can still afford the received item.

In cases where bidders have the same values and the same budget for an item
the auction might break ties randomly.
Also in this case, an agent would have an incentive to deviate from truthful bidding as long as his probability of winning the item is positive. 
\qed \endproof

\MB{The special case described at the end of Theorem 1 discusses random tie-breaking in an otherwise deterministic mechanism.
We point the interested reader to \citet{che2013assigning} for the analysis of randomized mechanisms in the presence of hard budget constraints.}

The set of Nash equilibria in the complete-information game would be for combinations of two of the bidders to report a budget of 2, which is not a truthful strategy. However, there is an equilibrium selection problem. If the bidders cannot coordinate and the third bidder also reported a budget of 2, the price would increase to 2. If all bidders report a budget of 2 and winners are selected randomly, they get an expected utility of negative infinity. Still, it is not a Nash equilibrium for all three bidders to report a budget of 1 only, because each bidder can deviate unilaterally. 

One might argue that, in practice, bidders will rather play the off-equilibrium strategy and report their budget truthfully. However, in practice, bidders might have a utility that is negative, but not negative infinity, when they exceed their budget, such that bidders could again have an incentive to misreport their budget.


\MB{In Section \ref{sec:dgs_auction}, we introduce an iterative auction that always finds an outcome in the core in markets with budget-constrained and truthful unit-demand buyers. If an additional condition on the preferences is satisfied, the auction is ex-post incentive-compatible and finds a strongly Pareto-optimal outcome. Without this additional condition, the auction finds an outcome in the core if bidders are truthful, but it cannot be incentive-compatible in general, as shown in Theorem \ref{thm:ic}, and it is not necessarily strongly Pareto-optimal.}

\hide{ 
Efficient algorithms for determining core outcomes under budget constraints have been discussed in the literature. However, desirable properties like bidder-optimality are only guaranteed if additional assumptions on the bidders' preferences are made. \cite{aggarwal} introduced the notion of \emph{general position}, a sufficient condition for ascending auctions to indeed find the surplus-maximizing core-stable outcome. 
Let us provide a brief discussion, because the condition has received attention in the literature:

\begin{definition}[\cite{aggarwal}]\label{def:genpos}
	Consider a directed bipartite graph with edges between bidders $\mcal{B}$ and goods $\mcal{S}$ (including dummy good $0$): For $i \in \mcal{B}$ and $j \in \mcal{S}$, there is a
	\begin{itemize}
		\item forward-edge from $i$ to $j$ with weight $-v_i(j)$
		\item backward-edge from $j$ to $i$ with weight $v_i(j)$
		\item maximum-price edge from $i$ to $j$ with weight $b^i - v_i(j)$
		\item terminal edge from $i$ to the dummy good $0$ with weight $0$.
	\end{itemize}
	The auction is in $\emph{general position}$, if for every bidder $i$, there are no two alternating walks, following alternating forward and backward edges and ending with a distinct maximum-price or terminal edge, having the same total weight.
\end{definition}

On the first question, \cite{aggarwal} introduced the notion of \emph{general position}, a sufficient condition for ascending auctions to indeed find the surplus-maximizing core-stable outcome. However, this general position is difficult to compute and communicate, and it is only a sufficient condition. This means that there are valuations that are not in general position, but the auction is still incentive-compatible. 
\cite{henzinger2015truthful} argue the general position condition is rather restrictive, as it excludes for example symmetric bidders. Besides, they show that no polynomial-time algorithm can determine whether a set of valuations is in general position. 

\begin{example}
Consider an auction with two bidders $1$ and $2$ with $b^1=b^2$. The number of goods and the bidders' valuations may be chosen arbitrarily. Assume $j \in \mcal{S}$ is any good. Consider the following path starting from bidder $1$: $1 \rightarrow j \rightarrow 2 \rightarrow j$, where the last edge is a maximum-price edge, with total weight $-v_1(j) + v_2(j) + (b^2-v_2(j)) = b^2-v_1(j)$. Now consider the path $1 \rightarrow j$, where the only edge is a maximum-price edge, with weight $b^1-v_1(j)$. Since $b^1 = b^2$, the total weight of both paths is equal, so the auction is not in general position.
\end{example}
} 

\section{An Iterative Auction} \label{sec:dgs_auction}

Our auction is based on the well-known auction by \citet{demange1986multi} (denoted as DGS auction from now on), which implements the Hungarian algorithm.
We will provide conditions when it is ex-post incentive-compatible to report the demand set truthfully. Thus, we provide a natural generalization of the DGS auction to markets where bidders have binding budget constraints. The simple ascending nature of our auction also naturally motivates an optimality condition for the returned allocation. Without loss of generality, we will assume $r_j = 0$ in this section.


In the auction process, it may happen that we need to ``forbid'' some bidder to demand a certain item. We model this by introducing subsets $R_1, \dots, R_n \subseteq \mcal{S}$ of goods for every bidder and define the \emph{restricted demand set} to be
\[
D_i(p,R_i) = \left\{j \in R_i\,:\, p(j) \leq b^i \, \wedge \pi_i(j,p) \geq \pi_i(k,p) \, \forall k\in R_i \text{ with } p(k) \leq b^i\right\}.
\]
The set consists of all affordable items that generate the highest utility among all items in $R_i$.
We introduce the well-known notions of over- and underdemanded sets \citep{demange1986multi,mishra2006oudemand}, adjusted to our notion of restricted demand sets.

\begin{definition}
	Let a price vector $p$ and sets $R_1,\dots,R_n \subseteq \mcal{S}$ with $R_i \neq \emptyset \,\ \forall i$ be given. A set $T \subseteq \mcal{S}$ is
	\begin{itemize}
		\item \emph{overdemanded}, if $0 \not\in T$ and $|\{ i \in \mcal{B} \,:\, D_i(p,R_i) \subseteq T\}| > |T|$, and
		\item \emph{underdemanded}, if $p(j) > 0$ for all $j \in T$ and $|\{i \in \mcal{B} \,:\, D_i(p,R_i) \cap T \neq \emptyset\}| < |T|$.
	\end{itemize}
$T$ is \emph{minimally over-/underdemanded}, if it does not contain a proper over-/underdemanded subset.
\end{definition}
Finally, we define the strict budget set of bidder $i$ by $B_i(p) = \{j \in \mcal{S}\,:\, p(j) < b^i\}$. It consists of all items with prices strictly less than the bidder's budget.

\subsection{The Auction Algorithm}

The Algorithm \ref{alg:auction} describes our auction from the perspective of the auctioneer. Ex-ante, the auctioneer does not have prior information about the values or budgets of bidders, just the demand sets that are revealed in each round. 
At the beginning of each round, the auctioneer collects the demand sets from all bidders. The algorithm then alternates between two tasks: identifying bidders whose shrinking demand sets may lead to underdemand, and resolving overdemanded sets of items.
In step 2, we check whether there exists a set of bidders $I^t$ who have reached their budget constraint for the current price vector. These bidders were expressing demand for the same set of items at the previous prices, but after the last increment have dropped out. 

If $I^t$ is nonempty, then, at step 3, we pick a bidder from the set, and prohibit the bidder from including that particular set of items in the demand set from this point on, while the demands of all other bidders remain unchanged. Additionally, prices for the aforementioned set of items are set back to the point before bidders hit their budget. After this modification, the algorithm jumps back to step 2 and requests the updated demand sets from all bidders. If $I^t$ is empty, and there is overdemand, prices for the minimally overdemanded set of items are raised by 1 and the auction goes back to step 2. A set of items is minimally overdemanded, if the number of bidders demanding
only items in this set is greater than the number of items in this set. Minimally, means that no proper subset of the items is overdemanded. 
If there are not at least two bidders hitting their budget at the same time in this round, and there is no overdemand, then the auction terminates at step 5, computing the final assignment of items to bidders and the current price vector.

The auction algorithm  is based on the following observation:

\begin{lemma}\label{lem:resticted_ce}
	An outcome $(\mu,p)$ is in the core if and only if there are sets $R_1,\dots,R_n \subseteq \mcal{S}$ such that $B_i(p) \subseteq R_i$ and $\mu(i) \in D_i(p,R_i)$ for all $i$.
\end{lemma}
\proof{Proof:}
	Suppose first that $(\mu,p)$ is a core outcome. Set $R_i = \{j \in \mcal{S}\,:\, p(j) < b^i\} \cup \{\mu(i)\}$. Then $\mu(i) \in D_i(p,R_i)$, since otherwise there would exist an item $j$ with $p(j) < b^i$ generating a higher utility than $\mu(i)$ - this would constitute a blocking pair.
	
	Now let us assume that there are sets $R_i$ as described with $\mu(i) \in D_i(p,R_i)$ for all $i$. Suppose there is a blocking pair $(i,j)$. Then item $j$ costs strictly less than $b^i$, so $j \in R_i$, and item $j$ generates a higher utility than $\mu(i)$. This would contradict $\mu(i) \in D_i(p,R_i)$. Thus, $(\mu,p)$ is a core outcome. \qed
\endproof

Computing a core outcome can thus also be interpreted as finding a ``competitive equilibrium'' with respect to the restricted demand sets $D_i(p,R_i)$.
In view of Lemma \ref{lem:resticted_ce}, the goal of our auction procedure is to determine prices $p$ together with sets $R_i$, such that there are neither over- nor underdemanded sets of items. As observed by \citet{mishra2006oudemand}, this implies existence of an assignment $\mu: \mcal{B} \rightarrow \mcal{S}$, such that every bidder receives an item in his demand set, and every item with positive price gets assigned to some bidder. The following result is due to \cite{mishra2006oudemand}.

\begin{proposition}
	\label{prop:oudem_assignment}
	Suppose that with respect to $D_i(p,R_i)$, there is no over- or underdemanded set of items. Then there is an assignment $\mu: \mcal{B} \rightarrow \mcal{S}$ such that $\mu(i) \in D_i(p,R_i)$ for all $i$, and for all $j \in \mcal{S}$ with $p(j) > 0$, there is some $i$ with $\mu(i) = j$.
\end{proposition}
Note that \cite{mishra2006oudemand} consider markets without budgets and demand sets without restrictions. However, their proof only uses combinatorial properties of the demand sets, so it can be directly adapted to our setting. Thus, we omit a proof here.
\begin{algorithm}
	\DontPrintSemicolon
	\caption{Iterative Auction}
	\label{alg:auction}
	\textsc{initialize:} Set $p^1 = (0,\dots,0)$ and $R^1_i = \mcal{S}$ for all bidders $i$. Set $t=1$, $O^0 = \emptyset$ and $I^1 = \emptyset$.\;
	\textsc{check demand:} Request $D_i(p^t,R^t_i)$ from all bidders. If $t > 1$ and the set 
	\[
	I^t= \{i \in \mcal{B}\,:\, D_i(p^{t-1},R_i^{t-1}) \subseteq O^{t-1} \, \wedge \, D_i(p^{t-1},R_i^{t-1}) \setminus D_i(p^t,R_i^t) \neq \emptyset\}
	\]
	is nonempty, go to Step 3. Otherwise, if there is an overdemanded set, go to Step 4. Else, go to Step 5.\;
	\textsc{restrict:} Choose a bidder $i \in I^t$ and define $J_i^t = D_i(p^{t-1},R_i^{t-1}) \setminus D_i(p^t,R_i^t)$. Set $R_i^{t+1} = R_i^t \setminus J_i^t$, $O^t = \emptyset$ and $p^{t+1} = p^{t-1}$. For all other bidders $i'$, the sets $R_{i'}^{t+1} = R_{i'}^t$ are not changed. Set $t=t+1$ and go to Step 2.\;
	\textsc{price increment:} Choose a minimally overdemanded set $O^t$. For all $j \in O^t$, set $p^{t+1}(j) = p^t(j)+1$. The prices for all other goods, as well as the sets $R_i^t$ remain unchanged. Set $t = t+1$ and go to Step 2.\;
	\textsc{return outcome:} Compute an assignment $\mu$, such that $\mu(i) \in D_i(p^t,R_i^t)$ for all bidders $i$ and $\mu(\mcal{B}) \subseteq \{j \in \mcal{S}\,:\, p^t(j) > 0\}$. Set $p = p^t$ and return $(\mu,p)$.\;
\end{algorithm}

Step 3 of the auction ensures that we do not end up with underdemanded sets of items. Each bidder carries around a growing exclusion set $R_i$, which addresses ``underdemand'' without ever querying new budget or value information. 
\MB{A bidder $i \in I^t$ can be chosen randomly if $I^t>1$, but there is no guarantee for strong Pareto optimality.  (Section 7 shows that even with full knowledge, finding the surplus-maximizing core outcome is NP-complete.) }
If we compute $p^{t+1} = p^{t-1}$, then the outcome is strongly Pareto-optimal, but it is not necessarily in the strong core. Example \ref{ex:strongPO_weak} makes this point.

The sets $R_i^t$ always contain at least all items that cost strictly less than the bidder's budget $b^i$. Our proof of correctness is similar to the one by \cite{laan2016ascending}: due to the budget constraints, underdemanded sets of items may appear. We show that Step 3 of the auction takes care of these sets.
\begin{lemma}\label{lem:overdem_not_underdem}
	Let $O$ be minimally overdemanded and $T \subseteq O$ with $T \neq \emptyset$. Let prices $p$ and sets $R_i$ be given. Then
	\[
	|\{i\,:\, D_i(p,R_i) \subseteq O \,\wedge \,D_i(p,R_i) \cap T \neq \emptyset\}| > |T|.
	\]
	In particular, $T$ is not underdemanded.
\end{lemma}
The proof of this lemma can be found in the Appendix.
\begin{lemma}\label{lem:restr_nonempty}
	For all bidders $i \in \mcal{B}$ and all iterations $t$ of the algorithm, we have that $B_i(p) \subseteq R_i^t$. In particular, since $p^t(0) = 0$, $R_i^t \neq \emptyset$.
\end{lemma}
\proof{Proof: }
	Assume to the contrary that there is a minimal iteration $t+1$, such that a bidder $i^*$ and a good $j^*$ exist with $p^{t+1}(j^*) < b^{i^*}$, but $j^* \not\in R_{i^*}^{t+1}$. Then in iteration $t$, Step 3 was executed, since otherwise $p^t \leq p^{t+1}$ and $R_{i^*}^t = R_{i^*}^{t+1}$, so $t+1$ would not be minimal. Hence, in iteration $t$, we have $j^* \in J_{i^*}^t$ and in particular $j^* \in D_{i^*}(p^{t-1},R_{i^*}^{t-1})$. Because Step 3 is executed, we have $O^{t-1} \neq \emptyset$, so in iteration $t-1$ Step 4 was executed and $p^t(j^*) = p^{t-1}(j^*)+1 = p^{t+1}(j^*) + 1 \leq b^{i^*}$. Thus, since from iteration $t-1$ to $t$, all prices for all preferred goods of bidder $i^*$ were raised and $i^*$ can still afford $j^*$ at prices $p^t$, $j^* \in D_{i^*}(p^t, R_{i^*}^t)$, so $j^* \not\in J_{i^*}^t$. This is a contradiction. \qed
\endproof

\begin{proposition}\label{prop:underdem_implies_3}
	For every iteration $t$ in the auction, it holds:
	\begin{enumerate}
		\item if there is a minimally underdemanded set of items $T$, then $T \subseteq O^{t-1}$ and Step 3 is executed
		\item if Step 3 is executed, there is no underdemanded set of items with respect to the $D_i(p^{t+1},R_i^{t+1})$.
	\end{enumerate}
\end{proposition}
\proof{Proof:} We prove this by induction on $t$. For $t = 1$, there clearly is no underdemanded set of items, and Step 3 is not executed.

Suppose now that $t > 1$ and that the statement is true for all $1 \leq s < t$. 

First suppose that there exists an underdemanded set of items $T$. Therefore, by induction, in iteration $t-1$, Step 4 must have been executed -- otherwise, there would not exist an underdemanded set. But then, using the same inductive reasoning, there was no underdemanded set in iteration $t-1$. It is thus easy to see that, since in iteration $t-1$ only prices for items in $O^{t-1}$ were raised, only the demand for those items could decrease, so $T$ must be a subset of $O^{t-1}$. By Lemma \ref{lem:overdem_not_underdem}, we have
\[
|\{i \in \mcal{B}\,:\, D_i(p^{t-1},R_i^{t-1}) \subseteq O^{t-1} \, \cap \, D_i(p^{t-1},R_i^{t-1}) \cap T \neq \emptyset\}| > |T|.
\]
In other words, since $|\{i \in \mcal{B}\,:\, D_i(p^{t},R_i^{t}) \cap T \neq \emptyset\}| < |T|$, there must be a bidder $i^*$ with $D_{i^*}(p^{t-1},R_{i^*}^{t-1}) \subseteq O^{t-1}$ and $D_{i^*}(p^{t-1},R_{i^*}^{t-1}) \cap T \neq \emptyset$, but $D_{i^*}(p^{t},R_{i^*}^{t}) \cap T= \emptyset$. This implies that $i^* \in I^t$, so Step 3 is executed in iteration $t$.

Now suppose that Step 3 is executed in iteration $t$.  Then again, in iteration $t-1$, Step 4 was executed, since otherwise we would have $O^{t-1} = \emptyset$, which implies $I^t = \emptyset$. By induction, there was no underdemanded set of items in iteration $t-1$. Note that $p^{t-1} = p^{t+1}$, so only the demand of a single bidder $i^* \in I^t$ chosen in Step 3 does change. Since $D_{i^*}(p^{t-1},R_{i^*}^{t-1}) \subseteq O^{t-1}$, so $J_{i^*}^t \subseteq O^{t-1}$, only the demand for items in $O^{t-1}$ can decrease. However, for $T \subseteq O^{t-1}$ we have again by Lemma \ref{lem:overdem_not_underdem} that
\[
|\{i \in \mcal{B}\,:\, D_i(p^{t-1},R_i^{t-1}) \subseteq O^{t-1} \, \cap \, D_i(p^{t-1},R_i^{t-1}) \cap T \neq \emptyset\}| > |T|,
\]
and, since we only changed $R_{i^*}^{t-1}$, the demand for items in $T$ can at most decrease by $1$. Thus, $T$ is not underdemanded in iteration $t+1$.
\qed \endproof

Employing the previous lemmata, we can proceed to prove correctness of our proposed auction.
\begin{proposition}\label{prop:alg_correctness}
The auction terminates after a finite number of iterations. Whenever Step 5 is reached, a valid assignment $\mu$ as described in that step exists. The resulting pair $(\mu,p)$ is a core outcome.
\end{proposition}
\proof{Proof:}
    Each execution of Step 3 removes at least one item from the restricted set $R_i^t$ of a bidder. Since the number of items is finite and no item can be removed twice, Step 3 can only be executed a finite number of times.
    Similarly, Step 4 increases the prices of items in a minimally overdemanded set. As prices increase monotonically and budgets are finite, eventually no item will be overdemanded. Therefore, Step 4 is also executed only finitely many times.
    Thus, in some iteration $t^*$, Step 5 is executed. By Lemma \ref{prop:underdem_implies_3}, there is no underdemanded set in iteration $t^*$, because otherwise Step 3 would have been executed. Similarly, there is no overdemanded set. Finally, because of Lemma \ref{lem:restr_nonempty}, no set $R_i^{t^*}$ is empty, so by Proposition \ref{prop:oudem_assignment}, an assignment $\mu$ as required exists. By Lemma \ref{lem:resticted_ce}, $(\mu,p)$ is a core outcome.
\qed \endproof

\begin{example}
In order to illustrate the rules, consider the following auction with three bidders $1,2,3$ and two items $A$ and $B$.
\begin{table}[H]
	\centering
	\begin{tabular}{c|c|c|c}
		& $v_i(A)$ & $v_i(B)$ & $b^i$ \\ \hline
		Bidder $i=1$ & $10$ & $0$ & $1$ \\
		Bidder $i=2$ & $0$ & $10$ & $2$ \\
		Bidder $i=3$ & $10$ & $10$ & $10$
	\end{tabular}
\end{table}
The auction proceeds as follows.
\begin{table}[H]
	\centering
	\begin{tabular}{c|c|c|c|c|c|c|c|c|c}
		& $p^t$& $D_1(p^t,R_1^t)$ & $D_2(p^t,R_2^t)$ & $D_3(p^t,R_3^t)$ & $R_1^t$ & $R_2^t$ & $R_3^t$ & $O^t$ & $I^t$ \\ \hline
		$t=1$ & $(0,0)$ & $\{A\}$ & $\{B\}$ & $\{A,B\}$ & $\mcal{S}$& $\mcal{S}$ & $\mcal{S}$ & $\{A,B\}$ & $\emptyset$ \\
		$t=2$ & $(1,1)$ & $\{A\}$ & $\{B\}$ & $\{A,B\}$ & $\mcal{S}$& $\mcal{S}$ & $\mcal{S}$ &$\{A,B\}$ & $\emptyset$ \\
		$t=3$ & $(2,2)$ & $\{0\}$ & $\{B\}$ & $\{A,B\}$ &$\mcal{S}$& $\mcal{S}$ & $\mcal{S}$ & $\emptyset$ & $\{1\}$ \\
		$t=4$ & $(1,1)$ & $\{0\}$ & $\{B\}$ & $\{A,B\}$ &$\{0,B\}$& $\mcal{S}$ & $\mcal{S}$ & $\emptyset$ & $\emptyset$
	\end{tabular}
\end{table}

In iterations $t=1,2$, there is a unique minimally overdemanded set $O^t = \{A,B\}$, and $I^t$ is empty. Thus, Step 4 of the auction is executed and the prices for $A$ and $B$ are raised. In iteration $t=3$, the set $I^t = \{1\}$ is nonempty which indicates that bidder 1's budget was tight for $A$ at prices $(1,1)$. Thus, we forbid $1$ to receive item $A$ and reset the prices to $(1,1)$. Now, in iteration $t=4$, there is no overdemanded set and $I^t$ is empty. Thus, there exists an assignment $\mu$ with $\mu(i) \in D_i(p^4,R_i^4)$ for all $i \in \mcal{B}$, namely $\mu(1) = 0$, $\mu(2) = B$ and $\mu(3) = A$. It is easily checked that $(\mu,p)$ is indeed a core outcome. 
\end{example}

\hide{
\begin{example}
	We now analyze an example of an auction that is not in general position. Consider the following setting.
	\begin{table}[h]
		\centering
		\begin{tabular}{c|c|c|c}
			& $v_i(A)$ & $v_i(B)$ & $b^i$ \\ \hline
			Bidder $i=1$ & $10$ & $0$ & $3$ \\
			Bidder $i=2$ & $0$ & $11$ & $1$ \\
			Bidder $i=3$ & $5$ & $3$ & $10$
		\end{tabular}
	\end{table}


It is easy to see that after $3$ iterations through Step 4 of our auction, we reach prices $p^4 = (3,1)$, where bidder $1$ demands $\{A\}$, bidder $2$ demands $\{B\}$ and bidder $3$ demands $\{A,B\}$. Since $\{A,B\}$ is minimally overdemanded, we execute Step $4$ once again to reach $p^5=(4,2)$, where due to the budget constraints, we have $D_1(p^5,R_1^5) = \{0\}$ and $D_2(p^5,R_2^5) = \{0\}$, while bidder $3$ still demands $\{A,B\}$. Thus, Step 3 of the auction is executed with $I^5 = \{1,2\}$, so both bidders $1$ or $2$ would be valid bidders to choose in Step 3. For the choice $i=1$, we have $J^5_i=\{A\}$, while for the choice $i=2$, we have $J^5_i = \{B\}$. We could thus either remove $A$ from $R^5_1$, or $B$ from $R^5_2$. Depending on our choice, we get two different core outcomes, both supported by the prices $p=(3,1)$: one where bidder $1$ receives nothing, bidder $2$ receives $B$ and bidder $3$ receives $A$, and one where bidder $1$ receives $A$, bidder $2$ receives nothing and bidder $3$ receives $B$.
\end{example}
}


Note that iterative auctions with demand queries require bidders to reveal that they are indifferent to not winning the good once the price equals the valuation of a bidder. Only this allows auctioneers to differentiate between a bidder dropping out due to reaching his valuation or his budget.  In practice, bidders might not always bid the null set when price reaches value, even in an ex-post incentive-compatible auction, which can lead to inefficiencies in any such iterative auction. \MB{The use of proxy agents as in \citet{ausubel2002ascending} solves this problem and makes the auction DSIC.}

\subsection{Economic Properties} \label{sec:pareto_welfare}

The output produced by our iterative auction is not uniquely defined - it may depend on which bidder $i \in I^t$ is chosen whenever Step 3 is executed. Indeed, we prove the following result.
\begin{proposition}\label{prop:reachable_alloc}
Let $(\nu,q)$ be an arbitrary core outcome. Then bidders $i \in I^t$ in Step 3 can be chosen in such a way, that for the resulting outcome $(\mu,p)$ we have that $p \leq q$ coefficient-wise, and $\pi_i(\mu(i),p) \geq \pi_i(\nu(i),q)$ for all bidders $i$.
\end{proposition}
\proof{Proof:}
	The proof can be found in the Appendix.
\\endproof

We say that a core outcome is $(\mu,p)$ \emph{Pareto-optimal for the bidders}, if for every core outcome $(\nu,q)$ with $\pi_i(\nu(i),q) > \pi_i(\mu(i),p)$ for some bidder $i$, there is a bidder $i'$ with $\pi_{i'}(\mu(i'),p) > \pi_{i'}(\nu(i'),q)$. Proposition \ref{prop:reachable_alloc} directly implies that for every core outcome which is Pareto-optimal for the bidders, there is an outcome $(\mu,p)$ reachable by the auction with $\pi_i(\mu(i),p) = \pi_i(\nu(i),q)$ for all $i \in \mcal{B}$.\footnote{\cite{aggarwal} prove that their algorithm for computing an outcome in the core always finds the bidder-optimal core outcome $(\mu,p)$, whenever the auction is in general position. Here, \emph{bidder-optimal} means that for every other core outcome $(\nu,q)$ we have that $\pi_i(\mu(i),p) \geq \pi_i(\nu(i),q)$ for all $i \in \mcal{B}$. Bidder-optimality thus implies Pareto optimality.  We show a similar result for our auction: if the bidder to choose in Step 3 of our auction is always unique, our auction also finds a bidder-optimal core outcome.}

\begin{corollary}\label{cor:bidderoptimal}
	Suppose that whenever Step 3 is executed, $|I^t| = 1$, i.e., there is a unique bidder to choose, and let $(\mu,p)$ be the uniquely determined outcome of the auction. Then for any core outcome $(\nu,q)$ we have that $p \leq q$ and $\pi_i(\mu(i),p) \geq \pi_i(\nu(i),q)$ for all bidders $i$, i.e., $(\mu,p)$ is bidder-optimal.
\end{corollary}
\proof{Proof:}
	Since $|I^t|=1$ in every iteration through Step 3, the outcome $(\mu,p)$ of the auction is unique. Proposition \ref{prop:reachable_alloc} now directly implies that for every core outcome $(\nu,q)$, we have $\pi_i(\mu(i),p) \geq \pi_i(\nu(i),q)$ and $p \leq q$.
\qed \endproof

\hide{ 
If the general position condition by \citet{aggarwal} is satisfied, it can be shown that $I^t$ never contains more than one bidder. Thus, our auction always finds a bidder-optimal outcome like the auction by \citet{aggarwal} that also needs value and budget queries.

\begin{definition}[\cite{aggarwal}]\label{def:genpos}
	Consider a directed bipartite graph with edges between bidders $\mcal{B}$ and goods $\mcal{S}$ (including dummy good $0$): For $i \in \mcal{B}$ and $j \in \mcal{S}$, there is a
	\begin{itemize}
		\item forward-edge from $i$ to $j$ with weight $-v_i(j)$
		\item backward-edge from $j$ to $i$ with weight $v_i(j)$
		\item maximum-price edge from $i$ to $j$ with weight $b^i - v_i(j)$
		\item terminal edge from $i$ to the dummy good $0$ with weight $0$.
	\end{itemize}
	The auction is in $\emph{general position}$, if for every bidder $i$, there are no two alternating walks, following alternating forward and backward edges and ending with a distinct maximum-price or terminal edge, having the same total weight.
\end{definition}
\begin{proposition}\label{prop:genpos_expost}
	Suppose the auction is in general position. Then in every iteration through Step 3 of our iterative auction, we have that $|I^t|=1$, and for the unique $i \in I^t$, we have $|J_i^t| = 1$.
\end{proposition}
\proof{Proof:}
	The proof can be found in the Appendix.
\qed \endproof

Note that our condition that $|I^t| = 1$ whenever Step 3 is reached is less demanding than the general position condition, since we do not require $|J_i^t| = 1$. Our condition is straightforward to check for the auctioneer when the auction is actually performed, and it is necessary for an ascending auction with demand queries to end in a surplus-maximizing outcome in the core with certainty. Without the condition being satisfied, i.e., if at least two bidders hit their budget at the same time, the auctioneer needs to guess the right bidder to exclude or otherwise surplus can be arbitrarily low.
} 

Let us now consider surplus-maximization properties of our auction. We first observe that a surplus-maximizing core outcome can always be found among the ones that are Pareto-optimal for the bidders.
\begin{proposition}\label{prop:welfmax_among_pareto}
	Let $(\mu,p)$ and $(\nu,q)$ be core outcomes. If $\pi_i(\mu(i),p) \geq \pi_i(\nu(i),q)$ for all bidders $i$, then
	\[
	\sum_{i \in \mcal{B}} v_i(\mu(i)) \geq \sum_{i \in \mcal{B}} v_i(\nu(i)).
	\]
\end{proposition}
\proof{Proof:}
	The proof can be found in the Appendix.
\qed \endproof

As we described above, by Proposition \ref{prop:reachable_alloc}, we can reach any core outcome which is Pareto-optimal for the bidders with our auction, and Proposition \ref{prop:welfmax_among_pareto} says that one of them must be surplus-maximizing. Now, if we always have $|I^t| = 1$, the outcome of our auction is unique, which proves our first main result.
\begin{theorem}
	Bidders in $I^t$ in Step 3 of the auction can be chosen such that the outcome of the auction is a surplus-maximizing core outcome.
	
	In particular, if $|I^t|=1$ whenever Step 3 is reached, the unique outcome of the auction is a surplus-maximizing core outcome. 
\end{theorem}

Note that perturbing the budgets is no solution as it randomizes the choice of the winner and can lead to arbitrarily low surplus. 
Let us emphasize that the algorithm does not require value or budget queries and it relies on demand queries based on prices only, which is similar to auction mechanisms used in the field.

Knowledge of the bidders' demand sets does not suffice in order to always choose the ``correct'' bidders in Step 3 to reach a surplus-maximizing outcome. Our hardness result in Section \ref{sec:value_queries} implies that even with perfect knowledge of the bidders' preferences, choosing the correct bidders in Step 3 is NP-hard. However, our condition $|I^t| = 1$ at least gives the auctioneer a simple certificate of optimality. \MB{Note that if budgets are drawn from a continuous distribution, the probability that two bidders have the same budget is zero. With the use of proxy agents, the bid increments can be very small such that even if the bid space is discretized, the probability of $|I^t| > 1$ is very small. }


Finally, we prove that if $|I^t|= 1$ whenever Step 3 is reached, then the auction is ex-post incentive-compatible. \MB{Again, with proxy agents, the mechanism becomes DSIC. }
Remember, Proposition \ref{prop:alg_correctness} proves that the auction terminates in a core outcome, which is weakly Pareto-optimal. 
In Corollary \ref{cor:bidderoptimal}, we show that whenever Step 3 is executed, if $|I^t| = 1$, then the outcome is a bidder-optimal assignment.
Let us first show that every feasible assignment is either in the core or has a blocking bidder-seller pair that involves a bidder who is not better off in this assignment than in the bidder-optimal assignment. Versions of the following lemma appeared in \citet{gale1985some} where it is attributed to J. S. Hwang. 

\begin{lemma}[Hwang’s Lemma, private budgets]\label{lemma:hwang}
Assume that in every auction round $t$, and \(
\lvert I^t\rvert \le 1 \quad\text{for all }t.\)
Let $(\mu,p)$ be a feasible assignment for the assignment market $M=(\mathcal B,\mathcal S, V, B, R)$.  Let $(\mu^*,p^*)$ be the bidder‐optimal assignment (with payoffs $\pi_i^*$), and define $\mathcal B^+ \;=\;\{\,i\in\mathcal B \mid \pi_i>\pi_i^*\,\}.$ If $\mathcal B^+\neq\varnothing$, then there exists a blocking pair $(k,j)\in(\mathcal B\setminus\mathcal B^+)\times\mathcal S$.
\end{lemma}


The proof can be found in appendix \ref{app:hwang}.

\begin{proposition}\label{prop:ic_auction}
	If, each time Step 3 is reached during the auction, we have \MB{$|I^t| \leq 1$,} then it is a (weakly) dominant strategy for each bidder $i$ to submit her demand set truthfully, hence the auction is ex-post incentive-compatible.
\end{proposition}

\proof{Proof:}
Proposition \ref{prop:ic_auction} follows from Lemma \ref{lemma:hwang}. If an assignment is not in the bidder-optimal core, then some bidder could indeed prefer a different assignment as the bidder-optimal one. However, Lemma \ref{lemma:hwang} asserts that, in such a case, the assignment would have a pair $(k,j)$ with a truthful bidder $k \in (\mcal{B}\backslash \mcal{B}^+)$, that would block this assignment also in the auction where bidders $i \in \mcal{B}^+$ report false demand sets. Suppose there does not exist such a blocking pair $(k,j)$. This means that no agent can improve his utility by false reporting. \qed
\endproof

\MB{In the version of Hwang’s Lemma where each bidder’s budget $b^i$ is private and potentially misreported, the argument in the critical “tight‐price” case (when $p(j)=p^*(j)=b^i$ and $\pi_{i'}=\pi^*_{i'}$) relies on the auction’s dynamics being “local” enough that no more than one bidder ever hits her budget cap in a single round.  We assume \( \lvert I^t\rvert \le 1 \) so that a unilateral deviation cannot trigger a chain of simultaneous budget‐induced price jumps or reassignments.  Below is a version of the lemma that assumes that budgets $b^i$ are fixed and public information. Here, every blocking argument reduces to a simple static comparison of prices versus known caps.  No control over simultaneous budget hits is required, and the $\lvert I^t\rvert\le1$ assumption is dropped, shown in the proof in appendix \ref{app:hwang_public}.}

\begin{lemma}[Hwang’s Lemma, public budgets]\label{lemma:hwang_public}
Let $M=(\mathcal B,\mathcal S,V,B,R)$ be an assignment market in which each bidder $i$ has a hard budget $b^i$ that is public information and cannot be misreported.  Let $(\mu^*,p^*)$ be the outcome of the ascending auction (with bidder payoffs $\pi_i^*$) where bidders bid truthfully, and let $(\mu,p)$ be any other feasible assignment (with payoffs $\pi$) where bidder $i$ misreports, and 
$ \mathcal B^+ = \{\,i\in\mathcal B \mid \pi_i>\pi_i^*\,\}.$
If $\mathcal B^+\neq\varnothing$, then there exists a blocking pair $(k,j)$ with $k\in\mathcal B\setminus\mathcal B^+$ and $j\in\mathcal S$.
\end{lemma}
With this adaptation of Hwang's Lemma to the case with public budgets, we can state a corresponding corollary to Proposition \ref{prop:ic_auction}.
\MB{
\begin{corollary}\label{corollary:ic_auction}
  Suppose that all budgets \(b^i\) are common knowledge.  Then, in the ascending‐price auction, each bidder’s weakly dominant strategy is to report her true demand set, and hence the mechanism is ex‐post incentive-compatible.
\end{corollary}
}

The next subsection shows that even though the ascending auction with public budgets but without $|I^t|\leq 1$ would is ex-post incentive-compatible, it does not necessarily find a surplus-maximizing core outcome.

\subsection{The Value of Value Queries}

If in some iteration of the above algorithm Step 3, is reached with $|I^t| > 1$, the ascending auction with only demand queries does not necessarily find the surplus-maximizing core outcome. If the auctioneer restricts items from the wrong bidder in the iterative auction, surplus can be arbitrarily low. A natural question to ask is whether one can avoid such losses in surplus, if the auctioneer had access to value queries during the auction and was able to elicit the true valuations of an item for the bidders who hit their budget constraint. For example, in Step 3 the auctioneer might always exclude the bidder with the lowest value for the item. Unfortunately, surplus of such an auction can still be arbitrarily low, as we show in example \ref{ex:lowrev}.



\begin{example}\label{ex:lowrev}
Suppose we have three bidders with the valuations for items $A$, $B$, and $C$ as described in Table \ref{tab:values}.
\begin{table}[h!]
	\centering 
\begin{tabular}{c|c|c|c|c}
	& $v_i(A)$ & $v_i(B)$ & $v_i(C)$ & $b^i$ \\ \hline
	Bidder $i=1$ & $6$ & $1$ & $0$ & $2$ \\
	Bidder $i=2$ & $5$ & $0$ & $2$ & $2$ \\
	Bidder $i=3$ & $0$ & $0$ & $M$ & $1$
\end{tabular}
\caption{Valuations of bidders in example \ref{ex:lowrev}}
\label{tab:values}
\end{table}

The first steps of the auction are as described in Table \ref{tab:firstrounds}.
\begin{table}[h!]
	\centering
\begin{tabular}{c|c|c|c|c|c|c|c|c|c}
	& $p^t$& $D_1(p^t,R_1^t)$ & $D_2(p^t,R_2^t)$ & $D_3(p^t,R_3^t)$ & $R_1^t$ & $R_2^t$ & $R_3^t$ & $O^t$ & $I^t$ \\ \hline
	$t=1$ & $(0,0,0)$ & $\{A\}$ & $\{A\}$ & $\{C\}$ & $\mcal{S}$& $\mcal{S}$ & $\mcal{S}$ & $\{A\}$ & $\emptyset$ \\
	$t=2$ & $(1,0,0)$ & $\{A\}$ & $\{A\}$ & $\{C\}$ & $\mcal{S}$& $\mcal{S}$ & $\mcal{S}$ & $\{A\}$ & $\emptyset$ \\
	$t=3$ & $(2,0,0)$ & $\{A\}$ & $\{A\}$ & $\{C\}$ & $\mcal{S}$& $\mcal{S}$ & $\mcal{S}$ & $\{A\}$ & $\emptyset$ \\
	$t=4$ & $(3,0,0)$ & $\{B\}$ & $\{C\}$ & $\{C\}$ & $\mcal{S}$& $\mcal{S}$ & $\mcal{S}$ & $\emptyset$ & $\{1,2\}$
\end{tabular}
\caption{First four rounds of the iterative auction in example \ref{ex:lowrev}.}
\label{tab:firstrounds}
\end{table}

Now we have to remove item $A$ either from $R_1^t$ or $R_2^t$. According to the greedy strategy (bidder 1 has a higher value for $A$ than bidder 2), we remove bidder 2 from $R_2^t$ and the auction proceeds as follows (see Table \ref{tab:lastrounds}).
\begin{table}[h!]
	\centering
\begin{tabular}{c|c|c|c|c|c|c|c|c|c}
	& $p^t$& $D_1(p^t,R_1^t)$ & $D_2(p^t,R_2^t)$ & $D_3(p^t,R_3^t)$ & $R_1^t$ & $R_2^t$ & $R_3^t$ & $O^t$ & $I^t$ \\ \hline
	$t=5$ & $(2,0,0)$ & $\{A\}$ & $\{C\}$ & $\{C\}$ & $\mcal{S}$& $\mathcal{S}\setminus \{A\}$ & $\mcal{S}$ & $\{C\}$ & $\emptyset$ \\
	$t=6$ & $(2,0,1)$ & $\{A\}$ & $\{C\}$ & $\{C\}$ & $\mcal{S}$& $\mathcal{S}\setminus \{A\}$ & $\mcal{S}$ & $\{C\}$ & $\emptyset$ \\
	$t=7$ & $(2,0,2)$ & $\{A\}$ & $\{0,C\}$ & $\{0,B\}$ & $\mcal{S}$& $\mathcal{S}\setminus \{A\}$ & $\mcal{S}$ & $\emptyset$ & $\{3\}$ \\
	$t=8$ & $(2,0,1)$ & $\{A\}$ & $\{C\}$ & $\{0,B\}$ & $\mcal{S}$& $\mathcal{S}\setminus \{A\}$ & $\mcal{S}\setminus \{C\}$ & $\emptyset$ & $\emptyset$
\end{tabular}
\caption{Last rounds of the iterative auction in example \ref{ex:lowrev}, if the bidder with the lowest value is excluded.}
\label{tab:lastrounds}
\end{table}

And the auction terminates with the assignment $\mu(1) = A$, $\mu(2) = C$, $\mu(3) = B$ at prices $p=(2,0,1)$ with surplus $6+2+0=8$.
On the other hand, if the auctioneer removes item $A$ from $R_1^t$ instead of $R_2^t$ in step $t=4$, the auction proceeds as in Table \ref{tab:altround}.

\begin{table}[h!]
	\centering
\begin{tabular}{c|c|c|c|c|c|c|c|c|c}
	& $p^t$& $D_1(p^t,R_1^t)$ & $D_2(p^t,R_2^t)$ & $D_3(p^t,R_3^t)$ & $R_1^t$ & $R_2^t$ & $R_3^t$ & $O^t$ & $I^t$ \\ \hline
	$t=5$ & $(2,0,0)$ & $\{B\}$ & $\{A\}$ & $\{C\}$ & $\mcal{S}\setminus\{A\}$& $\mathcal{S}$ & $\mcal{S}$ & $\emptyset$ & $\emptyset$
\end{tabular}
\caption{Last round of the iterative auction in example \ref{ex:lowrev}, if the bidder with the highest value is excluded.}
\label{tab:altround}
\end{table}

Now, the auction terminates with the assignment $\mu(1) = B$, $\mu(2) = A$, $\mu(3) = C$ with surplus $1+5+M$. Since $M$ is arbitrary, the loss of surplus can be arbitrarily high.
\end{example}


\section{Computational Complexity}\label{sec:value_queries}

We now analyze the complete-information case, where the auctioneer has access to all values and budgets and aims to find a surplus-maximizing core allocation with unit-demand bidders. This paper studies the computational complexity of the decision version of this problem, determining whether such an allocation exists. 
If the answer is negative, we cannot expect simple (polynomial-time) iterative auction mechanisms as used in practice to find this solution. 

The complexity class NP (nondeterministic polynomial time) consists of decision problems for which a given solution can be verified in polynomial time. The notion of NP-hardness is defined in terms of the ability to reduce any problem in NP to the problem in question. Since NP is defined around decision problems, NP-hardness proofs naturally start with decision versions.\footnote{Once the decision version's hardness is established, extending the result to the corresponding optimization problem usually follows naturally, since solving the optimization problem implies solving the decision problem multiple times to narrow down the optimal solution. We refer the interested reader to \citet{arora2009computational}.}

\begin{framed}
\noindent \textbf{Maximum Surplus Budget Constrained Stable Bipartite Matching (\textsc{MBSBM})}

\textbf{Input:} Two disjoint sets $\mcal{S}$ (sellers) and $\mcal{B}$ (buyers) of $n$ agents each, a budget $b^i$ for each agent $i \in \mcal{B}$ and a reserve value $r_j$ for each seller $j \in \mcal{S}$, a value $v_i(j)$ for each pair of agents $i \in \mcal{B}$ and $j \in \mcal{S}$, and a non-negative integer $k$

\textbf{Output:} Boolean value

\textbf{Question:} Does there exist a core outcome $\mu$ with total surplus $\sum_{i \in \mcal{B}} (v_i( \mu(i)) - r_{\mu(i)} )\geq k$?
\end{framed}

\begin{theorem}
	\label{thm:mbsbm-compl}
	\textsc{MBSBM} is NP-complete. 
\end{theorem}

The proof can be found in the appendix. If a decision problem is in NP and is NP-hard, then it is NP-complete. 
First, we show that the problem is in NP. In Appendix \ref{sec:milp},  we model the problem as a mixed-integer program. Once a problem is modeled as such and it is not exponential in the size of the problem description, there is a polynomial-time nondeterministic algorithm, where we guess the values of integer variables and solve the resulting linear program (LP) in polynomial time. The observation is important, because with more complex valuations such as in combinatorial exchanges with hard budget constraints of buyers, the problem was shown not to be in NP anymore and substantially harder to solve \citep{bichler2021or}. 
On the other hand, we know that for the case when all budgets are non-binding, $b^i \geq v_i(j)$ for all $i \in \mcal{B}$ and $j \in \mcal{S}$, the problem is equivalent to the maximum-weight bipartite matching, which admits a polynomial time solution via the Hungarian algorithm \citep{kuhn1955hungarian}. Core prices can be derived from the duals of the corresponding linear program \citep{shapley1971assignment}. Therefore, the case with unit-demand bidders with binding budgets is particularly interesting. 

Second, in Appendix \ref{sec:complexity}, we show that the problem is also NP-hard, i.e., it is at least as hard as the hardest problems in NP. The standard way to prove this is via a polynomial-time reduction from a known NP-hard problem to another one. Reductions transform one problem into another in such a manner that the transformation process (reduction) can be carried out in polynomial time, and the original problem has a solution if and only if the transformed problem has a solution.
The problem we reduce from is the Maximum Independent Set (\textsc{MIS}) problem, which is well-known to be both NP-hard and APX-hard, even under significant restrictions on the vertex degree of the graph \citep{garey}. The class APX consists of NP optimization problems that admit polynomial-time approximation algorithms capable of achieving a solution within a constant factor of the optimal. Specifically, APX contains problems for which an efficient algorithm can approximate the optimal solution to within a fixed multiplicative ratio. A problem is characterised as APX-hard, whenever there exists an $\epsilon > 0$, such that no $(1-\epsilon)$-approximation of the solution can be computed in polynomial time, unless P = NP. The dual classification of \textsc{MIS} as both NP-hard and APX-hard signifies that while you can approximate the solution within a constant factor, you cannot get arbitrarily close to the optimal solution in polynomial time. Through the polynomial-time reduction, these complexity characteristics are transferred to our problem, underscoring its inherent computational difficulty. 

\begin{framed}
\noindent\textbf{Maximum Independent Set (\textsc{MIS})}

\textbf{Input:} A graph $G = (V,E)$, with vertices $V$ and edges $E$, and a non-negative integer $k$

\textbf{Output:} Boolean value 

\textbf{Question:} Does there exist an Independent Set (IS) of size at least $k$, where as IS we define a set of vertices no two of which are adjacent?
\end{framed}

Our polynomial-time reduction uses a specific construction in which we introduce an individual vertex and an edge gadget for each vertex and edge of the maximum independent set (MIS) problem. In complexity theory, when performing a reduction from computational problem A to problem B, the term gadget refers to a subset of a problem instance of problem B that simulates the behavior of certain features of problem A. Drawing from graph theory, the vertex and edge gadgets are bipartite graphs where each edge gadget is connected to two vertex gadgets, corresponding to the two endpoints of the original edge. Each vertex in the gadget has a degree of three and thus each vertex gadget is connected to three edge gadgets. The edge gadget allows for two matchings between buyer and seller nodes, which all lead to the same surplus. The vertex gadgets also allow for two feasible stable matchings, where the surplus differs by one. The edges between an edge and a vertex gadget are such that it is not possible to select the surplus-maximizing matching in two consecutive vertex gadgets of two neighboring vertices, because it would generate pairs of blocking agents in the edge gadget of the connecting edge, i.e., a matching in the edge gadget would not be stable. Similar to the original MIS problem, where there cannot be two adjacent vertices in an independent set, in our construction, there cannot be two adjacent vertex gadgets with a high surplus matching. While in the MIS problem, we need to find the independent set that is maximal, in the MBSBM problem we need to determine the stable matching of the overall bipartite graph that maximizes surplus. Note that both the vertex and edge gadgets contain a pair of buyers that admit the same budget constraint, thus leading to a violation of the general position condition.


\section{Conclusions}

In his seminal paper, \citet{vickrey1961counterspeculation} showed that single-item auctions can be designed such that it is a dominant strategy for participants to reveal their preferences truthfully. The principle can be generalized \citep{varian1994generalized} to allow for general valuations and only poses a seemingly innocuous assumption of quasilinear and transferable utility. For single-item auctions, the Vickrey auction and the ascending (aka progressive or English) auction are incentive-compatible and implement the same outcome \citet{vickrey1961counterspeculation}. For multi-item auctions where bidders have unit demand, there also exists an ascending auction that implements the Vickrey outcome and is ex-post incentive-compatible \citep{demange1986multi}. Ascending auctions with general valuations are not ex-post incentive compatible unless we make strong additional assumptions \citep{de2007ascending}.

Budget constraints are widespread and they violate the transferable utility assumption in this literature on quasilinear auction design. It was already shown that in multi-unit auctions with quasilinear utilities and hard budget constraints, no incentive-compatible and efficient mechanism exists \citep{Dobzinski2008}. 
We show that in an assignment market with quasilinear and unit-demand preferences and hard budget constraints, there is no incentive-compatible mechanism terminating in a core solution for every input. Our ascending auction is strongly Pareto-optimal if a simple condition is satisfied, and in the core without this condition. If the condition holds, our auction is ex-post incentive-compatible. If the condition doesn't hold, the problem of finding a surplus-maximizing core outcome is NP-hard, and a simple polynomial-time auction (e.g., an ascending auction) cannot always find a strongly Pareto-optimal outcome. 



\bibliographystyle{ormsv080} 
\bibliography{literature} 

\newpage
\begin{APPENDICES}

\section{Proofs}

\subsection{Proofs from Section \ref{sec:dgs_auction}}
\proof{Proof of Lemma \ref{lem:overdem_not_underdem}.}
	Since $O \setminus T \subseteq O$ and $O$ is minimally overdemanded, $O \setminus T$ is not overdemanded, so
	\[
	|O|-|T| = |O \setminus T| \geq  |\{i\,:\, D_i(p,R_i) \subseteq O\setminus T \}|.
	\]
	Now
	\[
	\{i\,:\, D_i(p,R_i) \subseteq O\setminus T \} = \{i\,:\, D_i(p,R_i) \subseteq O\} \setminus \{i\,:\, D_i(p,R_i) \subseteq O \, \wedge \, D_i(p,R_i) \cap T \neq \emptyset\},
	\]
	so
	\[
	|O|-|T| \geq |\{i\,:\, D_i(p,R_i) \subseteq O\}| - |\{i\,:\, D_i(p,R_i) \subseteq O \, \wedge \, D_i(p,R_i) \cap T \neq \emptyset\}|.
	\]
	By rearranging terms we get
	\[
	|\{i\,:\, D_i(p,R_i) \subseteq O \, \wedge \, D_i(p,R_i) \cap T \neq \emptyset\}| \geq |\{i\,:\, D_i(p,R_i) \subseteq O\}| - |O| + |T|,
	\]
	and since $O$ is overdemanded, this implies
	\[
	|\{i\,:\, D_i(p,R_i) \subseteq O \, \wedge \, D_i(p,R_i) \cap T \neq \emptyset\}| > |T|.
	\]
\qed \endproof

\vspace{0.3cm}
In the following, we need a simple auxiliary lemma.

\begin{lemma}\label{lem:aux_overdem}
Let $(\nu,q)$ be an arbitrary outcome, and let $O \subseteq \mcal{S}$ be a minimally overdemanded set with respect to the demand sets $D_i(q,R_i)$. Let $\emptyset \neq J \subsetneq O$. Then there exists a bidder $i^*$ with $D_{i^*}(q,R_{i^*}) \subset O$, $D_{i^*}(q,R_{i^*}) \cap J \neq \emptyset$ and $\nu(i^*) \not\in J$.
\end{lemma}

\proof{Proof of Proposition \ref{prop:reachable_alloc}.}
We prove the following statements by induction on the iteration $t$, which imply Proposition \ref{prop:reachable_alloc}.
\begin{itemize}
	\item If in iteration $t$ Step 4 or 5 is executed, then $p^t \leq q$.
	\item If in iteration $t$ Step 3 is executed, then $p^{t+1} \leq q$ and we can choose bidder $i \in I^t$ such that $\nu(i) \not\in J^t_i$, and consequently, by choosing this bidder, we have $\nu(i) \in R_i^{t+1}$ for all $i$.
\end{itemize}

\noindent For iteration $t=1$ this is obviously true, since only Step 4 or Step 5 can be executed, and $p^1 = (0,\dots,0) \leq q$.

\noindent Now suppose that $t > 1$. First assume that in iteration $t$ Step 4 or 5 is executed. If in iteration $t-1$, Step 3 was executed, we have that $p^t = p^{t-2}$, and since Step 3 cannot be executed twice in a row, we have by induction $p^{t} = p^{t-2} \leq q$. Now assume that Step 4 was executed in iteration $t-1$. Then by induction we have $p^{t-1} \leq q$. Towards a contradiction, assume that the set $J = \{j \in O^{t-1} \,:\, p^t(j) > q(j)\}$ is not empty. Note that for all $j \in J$ we must have $q(j) = p^{t-1}(j)$. By Lemma \ref{lem:aux_overdem} there exists a bidder $i^*$ with $D_{i^*}(p^{t-1},R^{t-1}_{i^*}) \subset O^{t-1}$, $D_{i^*}(p^{t-1},R^{t-1}_{i^*}) \cap J \neq \emptyset$ and $\nu(i^*) \not\in J$. Thus, we have $\nu(i^*) \notin O^{t-1}$ or $p^t(j) \leq q(j)$. Since $\nu(i^*) \in R_{i^*}^{t-1}$, we have for every $j \in D_{i^*}(p^{t-1},R^{t-1}_{i^*}) \cap J$ that $\pi_{i^*}(\nu(i^*),q) < \pi_{i^*}(j,q)$. Consequently, since $(\nu,q)$ is a core outcome, so $(i^*,j)$ is no blocking pair, we must have that $q(j) = b^{i^*}$. But since $q(j) = p^{t-1}(j)$ and $j \in O^{t-1}$, it would follow that $j \notin D_{i^*}(p^t,R_i^t)$, so $I^t \neq \emptyset$ and Step 3 is executed in iteration $t$. This contradicts our assumption that Step 4 is executed.

Now consider the case where Step 3 is executed in iteration $t$. Since Step 3 cannot be executed twice in a row, we have by induction that $p^{t+1} = p^{t-1} \leq q$. It remains to show that there is some bidder $i \in I^t$ with $\nu(i) \notin J^t_i$. Again, towards a contradiction, assume that for all $i \in I^t$ we have $\nu(i) \in J_i^t$. Hence, $p^{t-1}(\nu(i)) = b^i$ for all $i \in I^t$, and since $q \geq p^{t-1}$ by induction, we have $q(\nu(i)) = b^i$ for all $i \in I^t$. Our argument is now very similar to the one above: Consider the set $J = \{j \in O^{t-1}\,:\,q(j) = p^{t-1}(j)\}$. Then $\nu(I^t) \subseteq J$, and in particular $J^t$ is nonempty. Hence there is a bidder $i^*$ with $D_{i^*}(q,R_{i^*}) \subset O$, $D_{i^*}(q,R_{i^*}) \cap J \neq \emptyset$ and $\nu(i^*) \not\in J$. Again, we have $\nu(i^*) \not\in O^{t-1}$, or $q(\nu(i^*)) > p^{t-1}(\nu(i^*))$. In both cases, $i^*$ would prefer any good $j$ in the intersection $D_{i^*}(q,R_{i^*}) \cap J \neq \emptyset$ to $\nu(i^*)$ at prices $q$. But since $i^* \not\in I^t$, $b^{i^*} > q(j)$, so $(i^*,j)$ would form a blocking pair.

\[
I = \{i \in \mcal{B}\,:\, D_i(p^{t-1},O^{t-1}) \subseteq O^{t-1} \, \wedge \, \nu(i) \in O^{t-1} \, \wedge \, p^{t-1}(\nu(i)) = q^{t-1}(\nu(i)) \}.
\]
Then $I^t \subseteq I$, so $I$ is not empty. By Lemma \ref{lem:aux_overdem}, there is a bidder $i^* \notin I$ with $D_{i^*}(p^{t-1},R_{i^*}^{t-1}) \subseteq O^{t-1}$ and $\nu(I)~\cap~D_{i^*}(p^{t-1},R_{i^*}^{t-1}) \neq \emptyset$. The $\nu(i^*) \notin O^{t-1}$, or $q(\nu(i^*)) > p^{t-1}(\nu(i^*))$. Since by induction we have $\nu(i^*) \in R_{i^*}^{t-1}$, we have that $\pi_{i^*}(\nu(i^*),q) < \pi_{i^*}(j,q)$ for all $j \in \nu(I) \cap D_{i^*}(p^{t-1},R_{i^*}^{t-1})$. But since $(\nu,q)$ is a core outcome, for no such $j$, $(i^*,j)$ can be a blocking pair -- implying that $b^{i^*} = q(j) = p^{t-1}(j)$. This is a contradiction, since it would follow that $i^* \in I$.

\qed \endproof

\proof{Proof of Proposition \ref{prop:welfmax_among_pareto}.}
	We have that
	\[
	\sum_{i \in \mcal{B}} v_i(\mu(i))  = \sum_{i \in \mcal{B}} \pi_i(\mu(i),p) + \sum_{j \in \mcal{S}} p(j) \geq \sum_{i \in \mcal{B}} \pi_i(\mu(i),p) + p(\nu(i))
	\]
	and
	\[
	\sum_{i \in \mcal{B}} v_i(\nu(i)) = \sum_{i \in \mcal{B}} \pi_i(\nu(i),q) + \sum_{j \in \mcal{S}} q(j) = \sum_{i \in \mcal{B}} \pi_i(\nu(i),q) + q(\nu(i)),
	\]
	since each good $j$ with $q(j) > 0$ is assigned to some bidder by $\nu$. Thus, it suffices to show that
	\[
	\sum_{i \in \mcal{B}} \pi_i(\mu(i),p)-\pi_i(\nu(i),q) + p(\nu(i))-q(\nu(i)) \geq 0.
	\]
	We show that each summand is non-negative by distinguishing two cases:
	
	Case 1: $p(\nu(i)) = b^i$. Then, since bidder $i$ receives $\nu(i)$ in the outcome $(\nu,q)$, $q(\nu(i)) \leq b^i = p(\nu(i))$. Since by assumption $\pi_i(\mu(i),p) \geq \pi_i(\nu(i),q)$, it follows that $\pi_i(\mu(i),p)-\pi_i(\nu(i),q) + p(\nu(i))-q(\nu(i)) \geq 0$.
	
	Case 2: $p(\nu(i)) < b^i$. Then $\pi_i(\mu(i),p) \geq \pi_i(\nu(i),p)$ - otherwise, $(i,\nu(i))$ would be a blocking pair with respect to prices $p$. Since $\pi_i(\nu(i),p) = \pi_i(\nu(i),q) + q(\nu(i)) - p(\nu(i))$, we again get that $\pi_i(\mu(i),p)-\pi_i(\nu(i),q) + p(\nu(i))-q(\nu(i)) \geq 0$.
\qed \endproof

\vspace{0.3cm}

\subsection{Proof of Hwang's Lemma with Private Budgets}\label{app:hwang}

The proof of Hwang's lemma follows those in \citet{demange1985some} and \citet{aggarwal}, but is adapted to our auction.

\proof{Proof of Lemma \ref{lemma:hwang}.}
Let us denote by $\mu(\mcal{B}^+)$, the set of items matched to bidders in $\mcal{B}^+$ in the assignment $\mu$, and by $\mu^*(\mcal{B}^+)$ the items matched to $\mcal{B}^+$ in the bidder-optimal matching with the truthful reports of all bidders $i \in \mcal{B}$. We consider two cases: $\mu(\mcal{B}^+) \ne \mu^*(\mcal{B}^+)$ (Case I) and  $\mu(\mcal{B}^+)=\mu^*(\mcal{B}^+)=\mcal{S}^+$ (Case II). We note that $\mcal{S}^+$ is defined analogously to $\mcal{B}^+$ and corresponds to the set of sellers that prefer the alternative assignment.

\vspace{0.3cm}

\textbf{Case I: $\mu(\mcal{B}^+) \ne \mu^*(\mcal{B}^+)$.}\\
For any $i \in \mcal{B}^+$ we have $\pi_i > \pi_i^* \geq 0$ and hence each bidder in $\mcal{B}^+$ is matched in $\mu$ to some item. By assumption, there exists an item $j \in \mu(\mcal{B}^+)$, $j \notin \mu^*(\mcal{B}^+)$. We define $i = \mu(j)$. Since $i \in \mu(\mcal{B}^+)$, $\pi_i > \pi_i^*$. 

By feasibility of the assignment, $p(j) \in [r_j, b^i]$ and $\pi_i + p(j) = v_i(j)$. We get three cases:
\begin{itemize}
	\item If $p^*(j) < b^i$ then $\pi_i^*+p^*(j) \geq v_i(j)$ because if $i$ wins $j$ then $\pi_i^*+p^*(j) = v_i(j)$. If $i$ wins some other item $j'$, then he has a higher utility for $j'$. Therefore, $p^*(j) \geq v_i(j)-\pi_i^* > v_i(j)-\pi_i = p(j) \implies \mathbf{p(j) < p^*(j)}$. 
	\item If $p^*(j) > b^i$ then $i$ would not win $j$ in $(\mu^*,p^*)$, but does so in $(\mu,p)$ hence $p(j) < b^i$, so we get $\mathbf{p(j) < p^*(j)}$.
  	\item If $p^*(j) = b^i$ and $p(j) \leq b^i$, then it can be that $\implies \mathbf{p(j) < p^*(j)}$, or $\mathbf{p^*(j)=p(j)}$. 
\end{itemize}
\vspace{0.3cm}

\textbf{For cases with $p(j)<p^*(j)$:} \\

Item $j$ is matched in $\mu^*$ to some $i'$, and by the choice of $j, i' \notin \mcal{B}^+$. Thus, $\pi_{i'} \leq \pi_{i'}^*$. 
By feasibility of $(\mu^*, p^*)$, $p(j)^* \in [r_{j}, b^{i'}]$ and $\pi^*_{i'} + p^*(j) = v_{i'}(j)$. 
Now, one can see that $(i',j)$ is blocking for $\mu$, because

\begin{align*}
& p(j) < p^*(j)\leq b^i\\
& \pi_{i'} \leq \pi^*_{i'} = v_{i'}(j) - p^*(j) \\ 
& \pi_{i'} + p(j) <  \pi^*_{i'} + p^*(j) = v_{i'}(j)\\
\end{align*}

Thus, $i'$ can pay more than $p(j)$ in the matching $(\mu,p)$ and thus pay more than what bidder $i$ can pay when he wins with a false report, and $i'$ can increase his utility $\pi_{i'}$ at the same time. This is a blocking pair in $(\mu,p)$.

\textbf{For case $p(j)=p^*(j)$:}\\
By assumption, $i$ and $i'$ cannot reach their budget at the same time in the truthful auction, i.e., $|I^t|\leq 1$, so $b^{i'} > b^i$, because $i'$ wins $j$ even though $i$ wants $j$.  So, if $\pi_{i'} < \pi_{i'}^*$, then $(i',j)$ is a blocking pair with respect to $(\mu,p)$.
\begin{align*}
& p(j) = p^*(j) = b^i\\
& \pi_{i'} < \pi^*_{i'} = v_{i'}(j) - p^*(j) \\ 
& \pi_{i'} + p(j) <  \pi^*_{i'} + p^*(j) = v_{i'}(j)\\
\end{align*}
Again, $(i',j)$ is a blocking pair for $(\mu,p)$.

The only remaining case: $p(j) = p^*(j) = b^i < b^{i'}$ and $\pi_{i'} = \pi_{i'}^*$.
\begin{align*}
& p(j) = p^*(j) = b^i\\
& \pi_{i'} = \pi^*_{i'} = v_{i'}(j) - p^*(j) \\ 
& \pi_{i'} + p(j) =  \pi^*_{i'} + p^*(j) = v_{i'}(j)\\
\end{align*}

\noindent Two possibilities arise:

\noindent 1. Can $i$ manipulate $b^i$? Lowering $\tilde b^{i} < b^i$ cannot increase the chances of winning $j$. Increasing $\tilde b^i > b^i$ might lead to winning, but at $p(j) > b^i$, so $\pi_i = -\infty < \pi_{i}^*$.

\noindent 2. Suppose $i$ reports $b^i$ truthfully, but reports valuations $\tilde v_i$ in order to win $j$. Consider for example 
\begin{align*}
\tilde v_i(k) &= 0 \quad \forall k \neq j
\\
\tilde v_i(j) &= \infty
\end{align*}

We will show that a losing bidder cannot become winning in this case unilaterally. 
If bidder $i$ wants to win item 
\MB{$j$}, he needs to demand this item up to his budget $b^i$. The only manipulation possible is to reduce demand on other items that he demands in the truthful auction. 
If the price of an item decreases due to the fact that bidder $i$ does not demand this item anymore, then the truthful bidders $i' \in (\mcal{B} - \mcal{B}^+)$ would still demand the item they win in the truthful auction and not demand additional items. As a consequence, the allocation would not change. So, this manipulation would not make bidder $i$ win item 
\MB{$j$}.

\vspace{0.3cm}
\textbf{Case II: $\mu(\mcal{B}^+)=\mu^*(\mcal{B}^+)=\mcal{S}^+$.}
We have that $0 \leq p(j) < p^*(j)$ for all $j \in \mathcal S^+$, since all bidders in $\mathcal B^+$ have a higher payoff than in the truthful auction $(\mu^*,p^*)$.
At the end of the truthful auction with outcome $(\mu^*,p^*)$, the prices of all items in $\mathcal S^+$ are strictly positive. If no bidder $i \not\in \mcal B^+$ was interested in an item $j \in \mcal S^+$, at least one item in $\mcal B^+$ would have a price of $0$. As long as the prices are less than $p^*(j)$, there has to be a bidder $i \not\in \MB{\mcal B^+}$ strictly preferring an item $j \in \mathcal S^+$. This bidder is blocking in the manipulated auction $(\mu,p)$.\qed
\endproof

\MB{\subsection{Proof of Hwang's Lemma with Public Budgets}\label{app:hwang_public} }
\proof{Proof of Lemma \ref{lemma:hwang_public}.}
Since each $i\in\mathcal B^+$ strictly prefers $(\mu,p)$ to the bidder‐optimal $(\mu^*,p^*)$, we have
$$
\pi_i > \pi_i^* \ge 0,
$$
so in $(\mu,p)$ every $i\in\mathcal B^+$ must be matched to some item.  Let $ S = \mu(\mathcal B^+)\setminus\mu^*(\mathcal B^+)$ be the set of items assigned in $(\mu,p)$ to bidders in $\mathcal B^+$ but not assigned to them in $\mu^*$.  
Every bidder in $\mathcal B^+$ has a match under both $\mu$ and $\mu^*$, so the fact that these matchings differ implies $\mu(\mathcal B^+) \ne \mu^*(\mathcal B^+)$, hence $S \ne \emptyset$. Pick any $j\in S$, and let $i=\mu(j)\in\mathcal B^+$.

By feasibility of $(\mu,p)$,
$$
p(j)\in[r_j,\,b^i], 
\quad
\pi_i + p(j) = v_i(j).
$$
Note that
\[
\pi_i + p(j) = v_i(j),
\quad
\pi_i^* + p^*(j) \ge v_i(j)
\quad\Longrightarrow\quad
p(j) \le p^*(j).
\]

We consider two cases:

\noindent\textbf{Case A:} $p(j) < p^*(j)$. Let $i'=\mu^*(j)$ be the bidder assigned to $j$ under $\mu^*$. 
Then
$$
\pi_{i'} + p(j)
= \bigl(v_{i'}(j)-p^*(j)\bigr) + p(j)
< v_{i'}(j),
$$
so bidder $i'$ strictly prefers item $j$ at price $p(j)$, and since $p(j) < p^*(j) \le b^{i'}$ he can afford it.  Thus $(i',j)$ is a blocking pair.

\noindent\textbf{Case B:} $p(j) = p^*(j)$ and $b^i<b^{i'}$. 
Feasibility of $(\mu,p)$ says that no winner can pay more than her budget, so $p(j)\leq b^i$. In the truthful auction $(\mu^*, p^*)$, bidder $i$ was unable to outbid the actual winner $i'$. Since a bidder with budget strictly above the final price could have outbid, we must have $b^i \leq p^*(j)$. Assume that ties are resolved in favor of the bidder with the larger budget such that $p(j)=p^*(j)=b^i<b^{i'}$.  Otherwise, bidder $i$ could have outbid $i'$ in $\mu^*$. 
Now:
\begin{itemize}
  \item If $\pi_{i'} < \pi^*_{i'}$, then
  $$
  \pi_{i'} + p(j)
  < \pi^*_{i'} + p^*(j)
  = v_{i'}(j),
  $$
  so $(i',j)$ blocks.
  \item If $\pi_{i'} = \pi^*_{i'}$, then since $b^{i'} > b^i$, bidder $i'$ could raise the price slightly to some $p' < b^{i'}$ (and $i'$ can afford it) with
  $$
    b^i < p' < b^{i'},
    \quad
    \pi_{i'} + p' = v_{i'}(j),
  $$
  improving his utility.  Again $(i',j)$ blocks.
\end{itemize}
In either case, we find a blocking pair $(i',j)$ with $i'\notin\mathcal B^+$.  

\noindent\textbf{Case C:} $p(j) = p^*(j)$ and $b^i=b^{i'}$. This case could not happen under the condition $|I^t|\leq 1$, but it can happen without the condition. 
If $i$ wins $j$ with $\pi_i>\pi_i^*$ and $i'$ wins some item $k$ with $\pi_{i'}>\pi_{i'}^*$, such an outcome would not be in the core. But the ascending auction always finds an outcome in the core such that this outcome cannot happen. 
If $i$ wins $j$ with $\pi_i>\pi_i^*$ and $i'$ wins some item $k$ with $\pi_{i'}=\pi_{i'}^*$. But in the ascending auction the bidders never reveal $v_i$ or $v_{i'}$. So, he cannot misreport to the auctioneer that $\tilde{v}_i(j)=\infty$. At the same time, bidding as if $\tilde{v}_i(k)=0$ for any other item $k$ cannot make bidder $i$ win $j$ either. If the price of an item decreases due to the fact that bidder $i$ does not demand this item anymore, then the truthful bidders $i' \in (\mcal{B} - \mcal{B}^+)$ would still demand the item they win in the truthful auction and not demand additional items.
Therefore, if in the truthful auction the auctioneer assigned item $j$ to $i'$ he will do so in the manipulated auction as well, because he has no other information. 


\endproof

\subsection{Proofs from Section \ref{sec:value_queries}}\label{app:proofs5}


\subsubsection{A MILP Formulation}\label{sec:milp}


First, we show that the problem of computing a surplus-maximizing (weak) core outcome belongs to the complexity class NP by modeling it as a Mixed Integer Linear Program (MILP).
Bilinear terms present in the quadratic formulation (weak-q-BC), namely products of continuous prices $p(j)$ and binary variables can easily be linearized via standard modeling tricks to obtain the resulting MILP. Note that the number of integer variables is not exponential in the size of the problem, and that after solving the LP, we can check that all constraints are satisfied in polynomial time such that the problem is in NP. 

\begin{align}
	\tag{weak-q-BC}
	\begin{array}{@{}l@{\quad}l@{\qquad}l@{}r@{}}
		\textrm{maximize} & \sum\limits_{i \in \mcal{B}} \pi_i + \sum\limits_{j \in \mcal{S}} \pi_j\\
		\textrm{subject to} & \pi_i = \sum\limits_{j \in \mcal{S}} (v_i (j) - p(j)) m_i(j) & \forall i \in \mcal{B}&\ (1)\\
		& \pi_j = \sum\limits_{i \in \mcal{B}} (p(j) - r_j ) m_i(j) & \forall j \in \mcal{S} &\ (2)\\
		& \sum\limits_{j \in \mcal{S}} m_i(j) \leq 1 & \forall i \in \mcal{B} &\ (3)\\
		& \sum\limits_{i \in \mcal{B}} m_i(j) \leq 1 & \forall j \in \mcal{S} &\ (4)\\
		& \pi_i \geq \big( v_i (j)  - p(j) \big) \alpha_i(j) & \forall i \in \mcal{B}, j \in \mcal{S} &\ (5)\\
		& \pi_j \geq \min(v_i (j),b^i)(1-y_i (j)) & \forall i \in \mcal{B}, j \in \mcal{S} &\ (6)\\
		& b^i \geq  p(j)  (1-\beta_i(j)) & \forall  i \in \mcal{B}, j \in \mcal{S} &\ (7)\\
		& p(j) \geq  b^i \beta_i(j) & \forall  i \in \mcal{B}, j \in \mcal{S} &\ (8)\\
		& (1-\alpha_i(j)) + (1-\beta_i(j)) -2 \leq 2(1- y_i (j)) + \epsilon y_i(j) & \forall i \in \mcal{B}, j \in \mcal{S} &\ (9)\\
		& r_j m_i(j) \leq p(j) \leq \min(v_i (j),b^i) m_i(j) + M(1-m_i(j)) &  \forall i \in \mcal{B}, j \in \mcal{S} &\ (10)\\
		& m_i(j) \in \{0,1\} &  \forall i \in \mcal{B}, j \in \mcal{S} &\ (11)\\
		& y_i (j) \in \{0,1\} & \forall i \in \mcal{B}, j \in \mcal{S} &\ (12)\\
		& \alpha_i(j) \in \{0,1\} &  \forall i \in \mcal{B}, j \in \mcal{S} &\ (13)\\
		& \beta_i(j) \in \{0,1\} &  \forall i \in \mcal{B}, j \in \mcal{S} &\ (14)\\
		& p(j) \geq 0 & \forall j \in \mcal{S} &\ (15)\\
	\end{array}\nonumber
	\label{weak-q-BC}
\end{align}

In this section, an assignment of buyer $i$ to seller $j$ is denoted as a binary variable $m_i (j)$. If the resulting assignment includes a pair $(i,j)$, then $m_i (j) = 1$, and for all other buyers except $i$, $m_{-i} (j) = 0$. The equivalence to the previous definitions is $m_i(j) = 1 \Leftrightarrow \mu(i) = j$. In order to check for the existence of deviating coalitions, two additional binary variables are introduced. Setting $\alpha_i(j) = 0$ represents the case where bidder $i$ has a benefit from deviating by trading with seller $j$, and $\alpha_i(j) = 1$ means that the bidder is best satisfied under the current assignment. The second auxiliary binary variable is set to $\beta_i(j) = 0$ if $i$ possesses a sufficient amount of money to purchase $j$, and set to $\beta_i(j) = 1$ if the budget of bidder $i$ is insufficient to acquire item of seller $j$, namely the set price of item $j$ exceeds $i$'s budget constraint. Variable $y_i(j) = 0$ reflects the case where bidder $i$ prefers to trade with seller $j$ and has sufficient budget, and the variable is set to 1 if one of the two necessary conditions does not hold. 

The utilities of buyers and sellers are defined as previously explained in Section 3. With $r_j$, we describe the reserve value or ask price of seller $j$. Constraints (1) and (2) represent the utilities of buyers and sellers, respectively. Constraints (5) and (6) guarantee (weak) core solutions. We examine all possible deviating combinations of buyer-seller pairs for a given outcome. Constraint (5) examines whether the corresponding payoff $\pi_i$ received by buyer $i$ in the selected assignment $m$ is higher or equal to the alternative assignment $(i,j)$ in question. In particular, this constraint checks whether an assignment $(i,j)$ yields a higher payoff for buyer $i$, in which case $\alpha_i(j) = 0$. Constraint (6) tests whether a seller $j$'s payoff $\pi_j$ on the optimal matching $m$ is higher or equal to the minimum value between any buyer $i$'s budget constraint $b^i$ and $i$'s valuation for the item $v_i(j)$, which represents the maximum possible payment seller $j$ could receive from any buyer. One or both of these conditions need to be true. Put differently, if both buyer and seller had a higher payoff under an alternative assignment $(i,j)$, outcome $m$ is not in the core. For an outcome to lie in the core, no buyer and seller pair must be able to both profit from an alternative allocation. These simultaneous constraints are captured by the logical \textit{and}. Since no alternative $(i,j)$ pair should increase their utilities, this translates to the logical \textit{or} of two negations: no bidder $i$ or no bidder $j$ can achieve a higher payoff in a different allocation. 

Constraints (7) and (8) examine whether bidder $i$ has a sufficient budget to obtain item $j$ under price $p(j)$. Constraint (9) is responsible for handling the value of $y_i(j)$ in an appropriate manner, to reflect whether a deviating coalition of $(i,j)$ is indeed profitable and budget-feasible, for any positive value $\epsilon<1$. The value of $y_i(j)$ depends on binary values $\alpha_i(j), \beta_i(j)$, and we verify our claim by examining the inequalities formed by the different value combinations (the tuple on the left side represents values $\alpha_i(j), \beta_i(j),y_i(j)$):
$$(0,0,0) \Rightarrow 1 + 1 - 2 \leq 2 \cdot 1 - \epsilon \cdot 0 \ \ (1^*)$$
$$(0,0,1) \Rightarrow 1 + 1 - 2 \leq 2 \cdot 0 - \epsilon \cdot 1 \ \ (2^*)$$
$$(0,1,1) \Rightarrow 1 + 0 - 2 \leq 2 \cdot 0 - \epsilon \cdot 1 \ \ (3^*)$$
$$(1,0,1) \Rightarrow 0 + 1 - 2 \leq 2 \cdot 0 - \epsilon \cdot 1 \ \ (4^*)$$
$$(1,1,1) \Rightarrow 0 + 0 - 2 \leq 2 \cdot 0 - \epsilon \cdot 1 \ \ (5^*)$$

In cases $( 1^*)$ and $(2^*)$, agent $i$ has a sufficient budget and can profit from deviating. However, inequality $(2^*)$ is infeasible, therefore the value of $y_i(j)$ cannot be set to 1 for this combination of $\alpha_i(j), \beta_i(j)$ and is forced to 0. For the remaining cases, either $i$ does not have sufficient budget ($\beta_i(j) = 1$), or has no profit from trading with $j$ ($\alpha_i(j) = 1$), or both conditions hold. In all the aforementioned cases, $y_i(j) = 1$, and thus reflects the case where no deviation is preferable from the buyer's side.

Constraint (10) then makes sure that if an item $j$ is assigned to buyer $i$, then the price is less than the minimum of the budget of this buyer or his value, and it is higher than the reserve price of the seller. 
We can conclude that the above formulation always results in the surplus-maximizing core outcome for assignment markets with budget constraints.

\subsubsection{Polynomial-Time Reduction from Maximum Independent Set}\label{sec:complexity}

Assume an instance of \textsc{MIS} defined on a cubic graph $G = (V, E)$, where $k$ represents the size of the set. A cubic graph is a graph in which all vertices have degree three. We define the transformed instance as a bipartite graph $G' = (V', E')$, with $V' = \mathcal{B} \cup \mathcal{S}$, and define functions $\pi_i:  \mathcal{S} \rightarrow \mathbb{R}_{\geq 0}$ for each buyer $i \in \mathcal{S}$, and $\pi_j: \mathcal{B} \rightarrow \mathbb{R}_{\geq 0}$ for each seller $j \in \mcal{S}$ that represent agents' payoffs.

\begin{itemize}
	\item $V'$ represents the total set of agents
	\item $\mathcal{B}$ and $\mathcal{S}$ denote the sets of buyers and sellers respectively
	\item $E'$ represents the potential transactions between buyers and sellers
	\item $\pi_i (j) := v_i(j) - b^i$ is the payoff of buyer $i \in \mathcal{B}$ from being matched to seller $j$
	\item $\pi_j (i) := b^i - r_j$ is the payoff of seller $j \in \mathcal{S}$ from being matched to buyer $i$ 
\end{itemize}

Since we assume that all assigned buyers pay prices for items equal to their budgets, we omit prices from the payoff formulae $\pi_i, \pi_j$. In Lemma \ref{lem:generality} we show that this restriction of the prices to be equal to the budgets is without loss of generality for our construction. Observe that in our specific construction, multiple buyers share the same budget constraint.

An \textit{assignment} $\mu: \mcal{B} \rightarrow \mcal{S}$ in $G'$ assigns each edge $(i,j)\in E'$ according to the condition $\sum_{j \in \mathcal{S}} \mathds{1}_{\{ \mu(i) = j \}} \leq 1$ for any agent $i \in \mathcal{B}$, where $\mathds{1}_{\{ \mu(i) = j \}}$ is the indicator function that equals 1 if $\mu(i) = j$ and 0 otherwise. The \textit{utility} of agent $i \in \mathcal{B}$ under assignment $\mu$ is defined as $\pi_i (\mu) := \pi_i(\mu(i))$ and similarly for agent $j \in \mathcal{S}$ it holds that $\pi_j (\mu) := \pi_j(\mu^{-1}(j))$. Given an assignment $\mu$, an edge $(i,j) \in E'$ is a \textit{blocking pair/ edge} if $\pi_i (\mu(i)) < \pi_i(j)$ and $\pi_j (\mu^{-1}(j)) < \pi_j(i)$. An assignment $\mu$ is \textit{stable} if it does not contain any blocking pair of agents. \\

\noindent\textbf{Edge Gadget.} Starting from an edge $e \in E$ of the original graph $G$, we construct the edge gadget $G'_e = \mathcal{B}_e \cup \mathcal{S}_e$ as a bipartite graph, with agent utilities for the subgraph defined as mentioned above with $\pi_i: \mathcal{S}_e \rightarrow \mathbb{R}_{\geq 0}$ and $\pi_j: \mathcal{B}_e \rightarrow \mathbb{R}_{\geq 0}$. Each vertex represents an \textit{agent}.

For each edge $e = (u, u') \in E$, we proceed to the following construction: 
\begin{itemize}
	\item Add three agents $\beta_e, \gamma_e, \delta_e \in \mathcal{B}_e$ 
	\item Add three agents $\eta_e, \alpha^u_e, \alpha^{u'}_e \in \mathcal{S}_e$
	\item Add two extra agents $\epsilon_u, \epsilon_{u'} \in \mathcal{B}_e$, if not present already, that will be part of the vertex gadget
\end{itemize}

Vertices $\alpha^u_e$ and $\alpha^{u'}_e$ act as gates to the vertex gadgets, which are connected to the vertices $\epsilon_u, \epsilon_{u'}$ of the vertex gadgets of $u, u'$ of the original graph. 

For each edge $e \in E$ of the original graph $G$, the corresponding edge gadget consists of the subgraph that agents $\{ \beta_e, \gamma_e, \delta_e \} \cup \{\eta_e, \alpha^u_e, \alpha^{u'}_e\}$ induce. 

If node $u \in V$ belongs to the independent set $IS(G)$, then we integrally match the pairs corresponding to the solid edges in Figure \ref{fig:edge_gadg}, otherwise, the dashed edges are matched. Any edge not present in the figure is assigned a value of zero. 

\begin{figure}[H]	
	\includegraphics[scale=0.32]{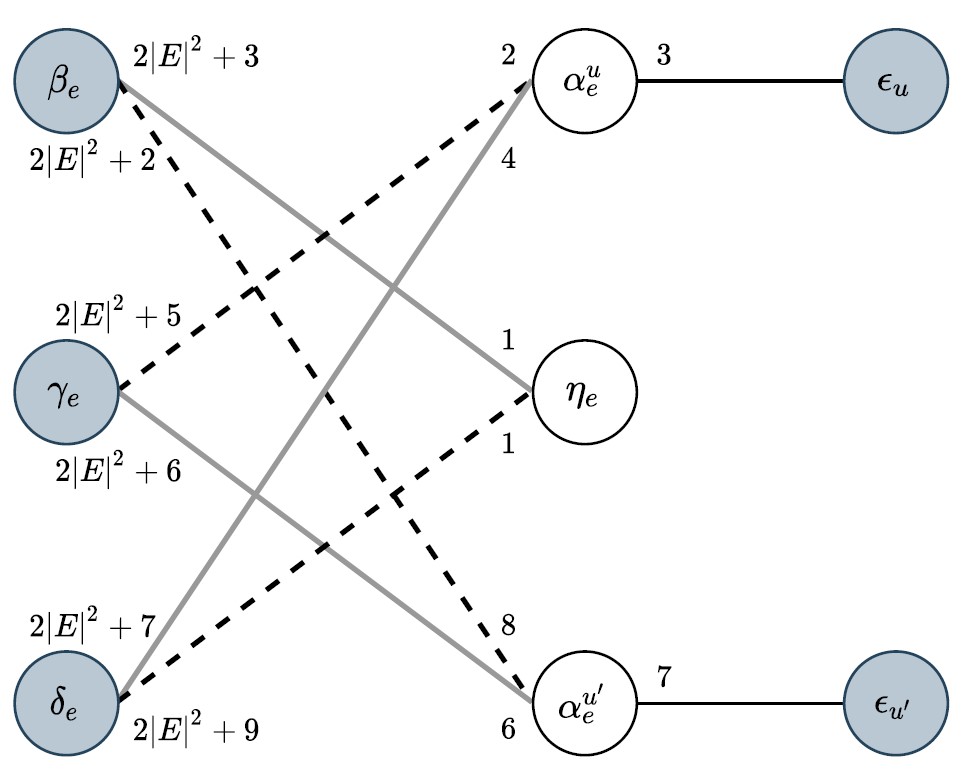}
    \caption{The edge gadget.}
	\label{fig:edge_gadg}
\end{figure}

An important observation is that any edge gadget $G'_e$ should contain two feasible stable assignments, with an equal total surplus. In every stable assignment $\mu$, at least one of $\{\alpha^u_e, \alpha^{u'}_e\}$ is unsatisfied under $\mu$, namely agents $\alpha^u_e$ or $\alpha^{u'}_e$ have a preference towards $\epsilon_u$ or $\epsilon_{u'}$ respectively. 

The valuations $v_i(j)$ of agents within the edge gadget are represented in Table \ref{tab:edge_values}, with $i \in \{\beta_e, \gamma_e, \delta_e, \epsilon_u, \epsilon_{u'}\}$ and $j \in \{\alpha^u_e, \eta_e, \alpha^{u'}_e\}$. Moreover $b^i$ corresponds to the budget of each buyer $i \in \mathcal{B}_e$ and $r_j$ to the reserve value of each seller $j \in \mathcal{S}_e$. The weights on each side of an edge, as captured in Figure \ref{fig:edge_gadg}, correspond to the buyer and seller utilities $\pi_i, \pi_j$.

\begin{table}[h!]
	\centering
	\begin{tabular}{l|llll}
		& $\alpha^u_e$   	& $\eta_e$   		& \multicolumn{1}{l|}{$\alpha^{u'}_e$}  	& $b^i$ \\ \hline
		$\beta_e$ 	& 0           	& $2|E|^2 + 12$ 	& \multicolumn{1}{l|}{$2|E|^2 + 11$}	& 9     \\
		$\gamma_e$ 	& $2|E|^2 + 12$ 	& 0          	& \multicolumn{1}{l|}{$2|E|^2 + 13$} 	& 7     \\
		$\delta_e$ 	& $2|E|^2 + 16$ 	& $2|E|^2 + 18$ 	& \multicolumn{1}{l|}{0}           	& 9     \\ 
		$\epsilon_u$	& 3				& 0				& \multicolumn{1}{l|}{0}			& 8		\\	
		$\epsilon_{u'}$& 0				& 0				& \multicolumn{1}{l|}{7}			& 8		\\	 \cline{1-4}
		$r_j$ & 5           & 8           & 1                                &      
	\end{tabular}
	\caption{Valuation table for the edge gadget.}
	\label{tab:edge_values}
\end{table}

\noindent\textbf{Vertex Gadget.}
We construct the vertex gadget corresponding to vertex $u \in V$ analogously. 

In Figure \ref{fig:vertex_gadg}, if a vertex $u$ belongs to the independent set $IS(G)$, we integrally match the pairs corresponding to the solid edges in the gadget $G'_u$, otherwise, the dashed edges are matched. The corresponding buyer and seller utilities are depicted on the weight of each side of the edge.

\begin{figure}[H]
	\includegraphics[scale=0.32]{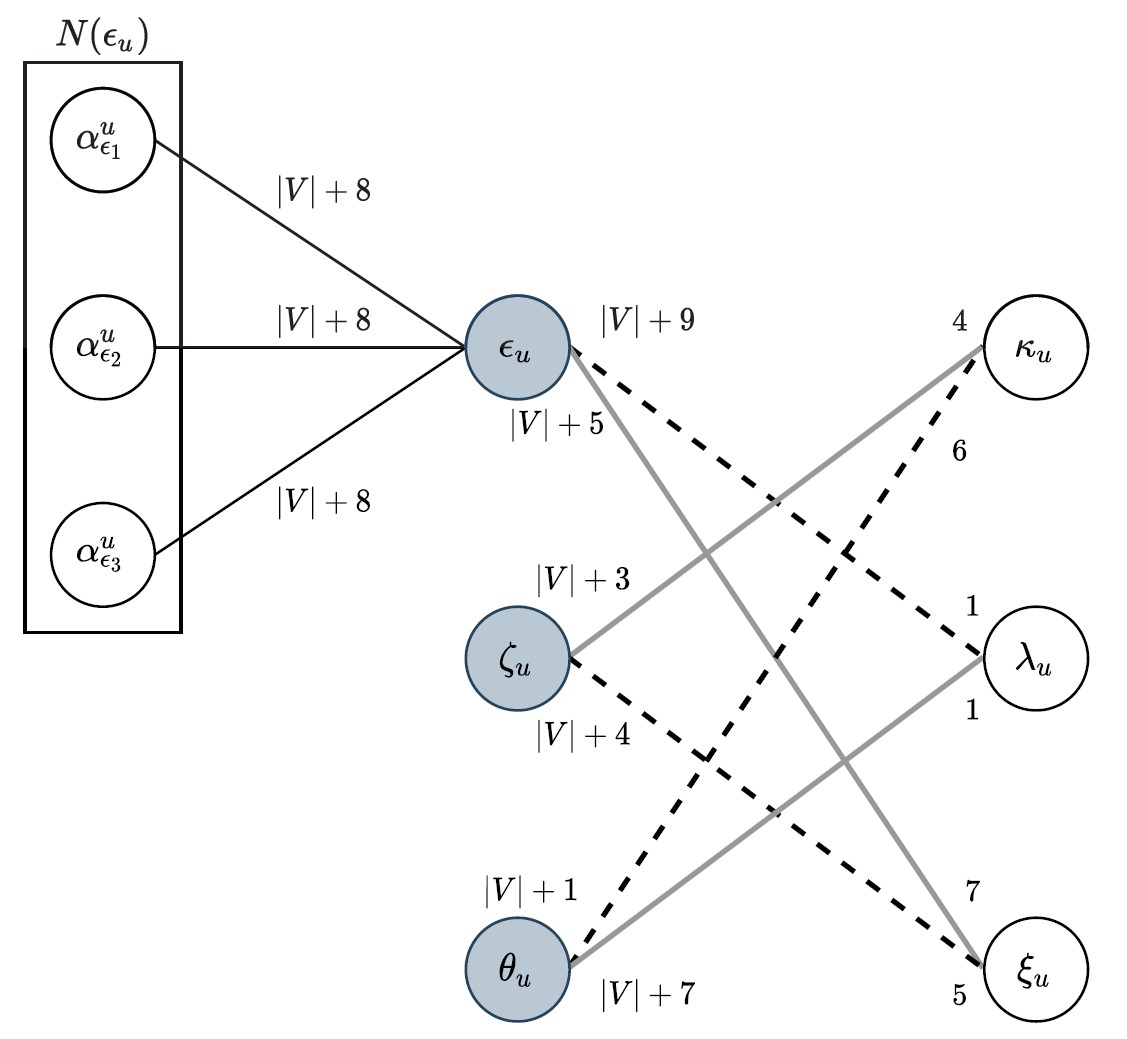}
	\caption{\small The vertex gadget.}
	\label{fig:vertex_gadg}
\end{figure}	

Table \ref{tab:vertex_values} represents the valuations of the vertices $v_i(j)$ of the vertex gadget, where $i \in \{\epsilon_u, \zeta_u, \theta_u\}$ and $j \in \{\kappa_u, \lambda_u, \xi_u, N(\epsilon_u)\}$. $N(\epsilon_u)$ consists of the neighboring vertices of $\epsilon_u$ in the edge gadgets of $\alpha^u_{\epsilon_1},\alpha^u_{\epsilon_2},\alpha^u_{\epsilon_3}$, since each vertex $u$ has a degree of three. 

\begin{table}[H]
	\centering
	\begin{tabular}{l|lllll}
		& $\kappa_u$ & $\lambda_u$    & $\xi_u$   & \multicolumn{1}{l|}{$N(\epsilon_u)$} & $b^i$ \\ \hline
		$\epsilon_u$ & 0     & $|V| + 17$ & $|V|+ 13$ & \multicolumn{1}{l|}{$n+16$}   & 8     \\
		$\zeta_u$ & $|V|+9$ & 0        & $|V|+10$  & \multicolumn{1}{l|}{0}        & 6     \\
		$\theta_u$ & $|V|+9$ & $|V|+15$   & 0       & \multicolumn{1}{l|}{0}        & 8     \\ \cline{1-5}
		$r_j$ & 2     & 7        & 1       &  5/1                          &      
	\end{tabular}
	\caption{Valuation table for the vertex gadget.}
	\label{tab:vertex_values}
\end{table}

Agents colored in gray in $G'_u$ belong to the set $\mcal{B}_u$, and agents in white belong to $\mcal{S}_u$, defined in an analogous manner to $\mcal{B}_e, \mcal{S}_e$. Two feasible stable outcomes exist in $G'_u$, resulting in surplus that differs by one: if the assignment indicated by the solid edges is chosen, the induced surplus equals $|V|+27$, while for the assignment of the dashed edges, surplus is equal to  $|V|+26$. The set $N(\epsilon_u)$ represents the neighborhood of $\epsilon$, and consists of three vertices belonging to three distinct edge gadgets.

Building towards the main result of Theorem \ref{thm:mbsbm-compl}, we show intermediate results captured by Lemmata \ref{lem:match-surplus} and \ref{lem:generality}. The first auxiliary lemma provides a useful bound for the aggregate surplus of an optimal assignment in an edge gadget.

\begin{lemma}
	\label{lem:match-surplus}
	Let $\mu$ be a stable assignment for edge gadget $G'_e$. The total surplus ($SW$) achieved by $\mu$ is at most $|E| \cdot ( 6|E|^2 +27 )$. If there is an edge $(u,u') \in E$, where it holds that $\pi_{\epsilon_u}(\mu (\epsilon_u)) < \pi_{\epsilon_u}(\alpha^u_e)$ and $\pi_{\epsilon_{u'}}(\mu (\epsilon_{u'})) < \pi_{\epsilon_{u'}} (\alpha^{u'}_e)$, then $SW < |E| \cdot (6|E|^2 +27) - |V|$.
\end{lemma}
\proof{Proof of Lemma \ref{lem:match-surplus}.}
	We examine the edge gadget $G'_e$ of edge $e = (u,u') \in E$. We define $\mcal{B}_e := \{\beta_e, \gamma_e, \delta_e\}$ as the set of buyers and $\mcal{S}_e := \{\alpha^u_e, \eta_e, \alpha^{u'}_e\}$ as the set of sellers. The total surplus that can be achieved by agents in $G'_e$ under assignment $\mu$ is $SW_e = \sum_{i \in \mcal{B}_e} \pi_i (\mu(i)) + \sum_{j \in \mcal{S}_e}\pi_j (\mu^{-1} (j))$. 
	
	For each individual edge gadget $G'_e$, we show that the maximum total surplus achieved is bounded by 
	$$SW_e \leq 3 \cdot 2|E|^2+27$$
	
	We prove the claim as follows:
	\begin{align*}
		SW_e &\ = \sum_{\substack{(i,j) \in \mcal{B}_e \times \mcal{S}_e \\ (i,j) \in E(G_e)}} \Big(\mathds{1}_{\{ \mu(i) = j \}} \cdot (\pi_i(j) + \pi_j (i)) \Big) +  \Big( \mathds{1}_{ \{ \mu(\epsilon_u) = \alpha^u_e \}} \cdot \pi_{\alpha^u_e} (\epsilon_u) \Big) +  \Big( \mathds{1}_{ \{ \mu(\epsilon_{u'}) = \alpha^{u'}_e \}} \cdot \pi_{\alpha^{u'}_e} (\epsilon_{u'}) \Big)\\
		&\ = \sum_{i \in \{\gamma_e, \delta_e\}} \Big( \mathds{1}_{ \{ \mu(i) = \alpha^u_e \}} \cdot (\pi_i(\alpha^u_e) + \pi_{\alpha^u_e} (i)) \Big) + \sum_{i \in \{\beta_e, \gamma_e\}} \Big( \mathds{1}_{ \{ \mu(i) = \alpha^{u'}_e \}} \cdot (\pi_i(\alpha^{u'}_e) + \pi_{\alpha^{u'}_e} (i)) \Big)\\
		&\ + \sum_{i \in \{\beta_e, \delta_e\}} \Big( \mathds{1}_{ \{ \mu(i) = \eta_e \}} \cdot (\pi_i(\eta_e) + \pi_{\eta_e} (i)) \Big) + \Big( \mathds{1}_{ \{ \mu(\epsilon_u) = \alpha^u_e \}} \cdot \pi_{\alpha^u_e} (\epsilon_u) \Big) +  \Big( \mathds{1}_{ \{ \mu(\epsilon_{u'}) = \alpha^{u'}_e \}} \cdot \pi_{\alpha^{u'}_e} (\epsilon_{u'}) \Big)\\
		&\ \leq 3 \cdot 2|E|^2 + 27
	\end{align*}
	
	We observe that each seller $\eta_e$ is connected to two buyer vertices inside the edge gadget $G'_e$, therefore, under assignment $\mu$, strictly one of $\mathds{1}_{\{ \mu(\beta_e) = \eta_e \}}$ and $\mathds{1}_{\{ \mu(\delta_e) = \eta_e \}}$ equals 1, while the other is 0. Sellers $\alpha^u_e$ and $\alpha^{u'}_e$ can be potentially matched with one of three different buyers: two of them inside $G'_e$, and one belonging to the vertex gadget, $u$ or $u'$ respectively. We compute the upper bound on the total sum of the valuations of sellers and buyers, on any possible assignment $\mu$, maintaining that, every time, one pair is matched, and the remaining conditions output 0 for all other edges connected to the pair. 
	The resulting social surplus $SW_e$ is equal to $3 \cdot 2|E|^2+27$ only when $\mathds{1}_{\{ \mu(\epsilon_u) = \alpha^u_e \}} = 0$ and $\mathds{1}_{\{ \mu(\epsilon_{u'}) = \alpha^{u'}_e \}} = 0$.

	We now proceed to bound the total surplus achieved under the condition that $\pi_{\epsilon_u} (\mu(\epsilon_u)) < \pi_{\epsilon_u}(\alpha^u_e)$ and $\pi_{\epsilon_{u'}} (\mu(\epsilon_{u'})) < \pi_{\epsilon_{u'}}(\alpha^{u'}_e)$ hold simultaneously. Since $\mu$ is a stable outcome, from the first condition, we derive that 
	\begin{equation}
		\pi_{\alpha^u_e} (\mu^{-1} (\alpha^u_e)) \geq \pi_{\alpha^u_e} (\epsilon_u) = 3
	\end{equation}
	Similarly, from the second condition, we get that
	\begin{equation}
		\pi_{\alpha^{u'}_e} (\mu^{-1} (\alpha^{u'}_e)) \geq \pi_{\alpha^{u'}_e} (\epsilon_{u'}) = 7
	\end{equation}
	
	In order to ensure that Equation (3) holds, assignment $\mu$ must force node $\alpha^u_e$ to be matched with $\delta_e$, as $\pi_{\alpha^u_e} (\delta_e) = 4 > 3$ in this case. If the pair $\{\gamma_e, \alpha^u_e \}$ was matched instead, it would hold that $\pi_{\alpha^u_e} (\gamma_e)  = 2 < 3$, thus contradicting condition (3).
	
	In a similar manner, $\mu$ matches pair $\{\beta_e, \alpha^{u'}_e\}$, achieving utility $\pi_{\alpha^{u'}_e} (\beta_e)  = 8 > 7$.
	
	The total surplus of assignment $\mu$ is therefore:
	\begin{align*}
		SW_e &\ = \sum_{\substack{(i,j) \in \mcal{B}_e \times \mcal{S}_e \\ (i,j) \in E(G_e)}} \Big( \mathds{1}_{ \{ \mu(i) = j \}} \cdot (\pi_i(j) + \pi_j (i)) \Big) +  \Big( \mathds{1}_{ \{ \mu(\epsilon_u) = \alpha^u_e \}} \cdot \pi_{\alpha^u_e} (\epsilon_u) \Big) +  \Big( \mathds{1}_{ \{ \mu(\epsilon_{u'}) = \alpha^{u'}_e \}} \cdot \pi_{\alpha^{u'}_e} (\epsilon_{u'}) \Big)\\
		&\ = \mathds{1}_{ \{ \mu (\delta_e) = \alpha^u_e \}} \cdot (\pi_{\delta_e} (\alpha^u_e) + \pi_{\alpha^u_e} (\delta_e)) +  \mathds{1}_{ \{ \mu (\beta_e) = \alpha^{u'}_e \}} \cdot (\pi_{\beta_e} (\alpha^{u'}_e) + \pi_{\alpha^{u'}_e} (\beta_e)) \\
		&\ \leq 2 |E|^2 + 11 + 2 |E|^2 + 10\\
		&\ \leq 4|E|^2 + 22 \\
		&\ < 6|E|^2 + 27 - |V|
	\end{align*}
	since, without loss of generality, we can assume that $2|E|^2 - |V| + 5 > 0$, as we are referring to cubic graphs, where $3 |V| = 2 |E|$ and as a  result the claim trivially holds.	 
\qed \endproof

This lemma implies that, in any surplus-maximizing assignment, edges between $\{\alpha^u_e, \epsilon_u\}$ and $\{\alpha^{u'}_e, \epsilon_{u'}\}$ are never chosen. Furthermore, at most one of the conditions $\pi_{\epsilon_u}(\mu (\epsilon_u)) < \pi_{\epsilon_u}(\alpha^u_e)$ and $\pi_{\epsilon_{u'}}(\mu (\epsilon_{u'})) < \pi_{\epsilon_{u'}} (\alpha^{u'}_e)$ must hold, otherwise assignment $\mu$ does not maximize surplus.

In our construction, we fix the price for each matching to the budget constraint of the buyer to whom a seller is matched. 

This restriction is without loss of generality. We show that, in each bipartite graph, regardless of whether it belongs to the vertex or edge gadget, the current scheme yields all possible surplus-maximizing outcomes in the core.

\begin{lemma}
	\label{lem:generality}
	Assuming that all buyers pay a price equal to their budget for their assigned items does not impact generality, namely, there does not exist a surplus-maximizing outcome in the core that is not reachable through this pricing scheme.
\end{lemma}

\vspace{0.3cm}
In order to prove Lemma \ref{lem:generality}, we provide an intermediate result regarding the existence of blocking pairs under specific assignments.

\begin{lemma}
	\label{lem:match-conditions}
	Consider edge $e=(u,u')$ of $G$, and let $\mu$ be an assignment for edge gadget $G'_e$. The following two statements hold:
	\begin{enumerate}
		\item If assignment $\mu$ satisfies $\mu(\beta_e)= \eta_e$, $\mu(\gamma_e)= \alpha^{u'}_e$ and $\mu(\delta_e) = \alpha^u_e $ (solid edges in Figure \ref{fig:edge_gadg}), then no blocking pair of $\mu$ involves any agent from $\{\beta_e, \gamma_e, \delta_e, \alpha^u_e, \eta_e\}$.
		\item If assignment $\mu$ satisfies $\mu(\beta_e) = \alpha^{u'}_e$, $\mu(\gamma_e) = \alpha^u_e$ and $\mu(\delta_e) = \eta_e$ (dashed edges in Figure \ref{fig:edge_gadg}), then no blocking pair of $\mu$ involves any agent from $\{\beta_e, \gamma_e, \delta_e, \alpha^{u'}_e, \eta_e\}$.
	\end{enumerate}
\end{lemma}
\proof{Proof of Lemma \ref{lem:match-conditions}.}
	For the proof of the first statement, we assume outcome $\mu$ sets values as suggested. Computing the utilities under $\mu$ for each agent, we get:
	$$\pi_{\beta_e}(\eta_e) = 2|E|^2 + 3,\ \pi_{\gamma_e} (\alpha^{u'}_e) = 2|E|^2 + 6,\ \pi_{\delta_e} (\alpha^u_e) = 2|E|^2 + 7$$
	$$\pi_{\alpha^u_e} (\delta_e) = 4,\ \pi_{\eta_e} (\beta_e) = 1$$
	One can easily verify that agents $\beta_e, \gamma_e, \alpha^u_e$ are maximizing their utility under $\mu$ by inspecting the valuation table in Table \ref{tab:edge_values}. Since agent $\eta_e$ is indifferent between agents $\delta_e$ and $\beta_e$, she has no incentive to deviate by forming a blocking pair with agent $\delta_e$. Therefore, there does not exist a blocking pair that includes agents $\{\beta_e, \gamma_e, \delta_e, \alpha^u_e, \eta_e\}$. The proof of the second statement follows similarly.
\qed \endproof 

\vspace{0.3cm}

 With the above lemma at hand, we can proceed to show Lemma \ref{lem:generality}.
\proof{Proof of Lemma \ref{lem:generality}.}
	We examine the edge and vertex gadgets separately, and argue that, in both cases, setting prices equal to the winning bidders' budgets yields all feasible surplus-maximizing outcomes in the core. Formally, the set of all surplus-maximizing assignments in the core coincides with the set of optimal assignments when prices are set at the budget limit.  
	
	In a two-sided matching, the surplus is defined as the gains from trade, the value of the buyers minus that of the sellers. This means that for each match between buyer $i$ and seller $j$, the corresponding surplus is computed as the sum of buyer $\pi_i(j) = v_i(j) - p(j)$ and seller payoff $\pi_j(i)= p(j) - r_j$, therefore prices are not included in the final sum, which is the result of the difference between assigned items' valuations and seller reserve values. 
	
	We begin by analyzing the edge gadget, as seen in Figure \ref{fig:edge_gadg}. The values depicted in Table \ref{tab:edge_values} represent the true valuation of each buyer among $\{\beta_e, \gamma_e, \delta_e\}$ for the item of each seller among $\{\alpha^u_e, \eta_e, \alpha^{u'}_e\}$.	One can trivially observe that there exist 6 feasible assignments between buyers and sellers in the bipartite graph. As stated above, prices do not participate in the surplus computation, and thus we can calculate the surplus-maximizing assignment based on the buyer valuation table. Since valuations $v_{\beta_e} (\alpha^u_e) = v_{\gamma_e} (\eta_e) = v_{\delta_e} (\alpha^{u'}_e)  = 0$, only 2 among 6 assignments are stable, and simultaneously surplus maximizing. The two assignments $\mu_1, \mu_2$ are $\mu_1(\beta_e) = \eta_e, \mu_1(\gamma_e) = \alpha^{u'}_e, \mu_1 (\delta_e) = \alpha^u_e$ and $\mu_2(\beta_e) = \alpha^{u'}_e, \mu_2(\gamma_e) = \alpha^u_e, \mu_2 (\delta_e) = \eta_e$. The total surplus admitted by assignments $\mu_1, \mu_2$ is equal to $6|E|^2 + 27$, while the remaining feasible assignments yield a strictly lower surplus. Lemma \ref{lem:match-conditions} states that assignments $\mu_1, \mu_2$, corresponding to the grey and dotted edges respectively, with seller payoffs set as $\pi_{\alpha^u_e} (\delta_e) = b_{\delta_e} - r_{\alpha^u_e} = 9 - 5 = 4$, $\pi_{\eta_e}(\beta_e) = b_{\beta_e} - r_{\eta_e} = 9 - 8 = 1$, $\pi_{\alpha^{u'}_e} (\gamma_e) = b_{\gamma_e} - r_{\alpha^{u'}_e} = 7 - 1 = 6$ under assignment $\mu_1$, and $\pi_{\alpha^u_e}(\gamma_e) = b_{\gamma_e} - r_{\alpha^u_e} = 7 - 5 = 2$, $\pi_{\eta_e}(\delta_e) = b_{\delta_e} - r_{\eta_e} = 9 - 8 = 1$, $\pi_{\alpha^{u'}_e}(\beta_e) = b_{\beta_e} - r_{\alpha^{u'}_e} = 9 - 1 = 8$ under assignment $\mu_2$, do not generate any blocking pairs. Thus, setting prices for matched buyers in the edge gadget equal to their budgets produces every feasible core-stable, surplus-maximizing assignment, and lowering prices cannot yield different stable assignments of higher or equal surplus. 
	
	In a similar analysis, we observe that, for the vertex gadget of Figure \ref{fig:vertex_gadg}, the surplus-maximizing assignment $\mu_3$, based on valuations described on Table \ref{tab:vertex_values}, is $\mu_3 (\epsilon_u) = N(\epsilon_u),\mu_3 (\zeta_u) = \xi_u, \mu_3 (\theta_u) = \lambda_u$. However, as suggested in Lemma \ref{lem:match-surplus}, assignment $\mu_0 (\epsilon_u)$ results in sub-optimal surplus for the overall subgraph including vertex and edge gadget. Therefore, the surplus-maximizing assignment $\mu_4$ is a result of an assignment between the vertices within the gadget. The aforementioned assignment is $\mu_4 (\epsilon_u) = \xi_u,\mu_3 (\zeta_u) = \kappa_u, \mu_3 (\theta_u) = \lambda_u$ (solid assignment in Figure \ref{fig:vertex_gadg}), admitting surplus equal to $3|V|+27$. In Lemma \ref{lem:match-conditions}, we have shown that assignment $\mu_4$, for winning buyer prices equal to their budgets is core-stable. Thus, the initial claim is true for each vertex gadget, concluding the proof of this lemma.
\qed \endproof

\vspace{0.3cm}

We are now ready to prove our main result.

\proof{Proof of Theorem \ref{thm:mbsbm-compl}.}
	Let $G = (V,E)$ be a cubic graph, with sizes of vertex and edge sets defined as $|V|$ and $|E|$ respectively. Since $G$ is assumed to be cubic, it does not possess any isolated vertices. An instance of \textsc{MIS} is defined by $G$ and an integer $k$. 
	
	A key property of cubic graphs is that all nodes must have a degree of 3. According to the \textit{Handshaking lemma}, it holds that $\sum_{u \in V} \text{deg}(u) = 2 |E|$, namely the sum of degrees of all vertices of the graph is twice as large as the size of the edge set. 
	
	We construct an instance $<G', \pi_i, \pi_j, SW>$ of \textsc{MBSBM}, where $G' = (V', E')$ is a bipartite graph, with $SW = |V| \cdot (3 |V| +26) + k + |E| \cdot (6|E|^2 + 27)$.  For each edge $e \in E$, the buyer and seller sets of the corresponding edge gadget are defined as $\mcal{B}_e, \mcal{S}_e$. For each vertex $u \in V$, the respective sets for the corresponding vertex gadget are defined as $\mcal{B}_u, \mcal{S}_u$. Unifying for all edge gadgets in $G'$, we define $\mcal{B}_E = \bigcup\limits_{e \in E} \mcal{B}_e$,  $\mcal{S}_E = \bigcup\limits_{e \in E} \mcal{S}_e$, and for all vertex gadgets, sets $\mcal{S}_V = \bigcup\limits_{u \in V} \mcal{S}_u$ and $\mcal{B}_V = \bigcup\limits_{u \in V} \mcal{B}_u$. The vertex set $V'$ consists of the union over all vertex and edge gadgets, and therefore $\mcal{B} = \mcal{B}_E\ \cup \mcal{B}_V $ and $\mcal{S} =   \mcal{S}_E \ \cup \  \mcal{S}_V$ represent the total number of buyers and sellers, respectively. Each edge gadget corresponding to edge $e \in E$ consists of vertices $\mcal{B}_e = \{\beta_e, \gamma_e, \delta_e \}$ and $\mcal{S}_e = \{\alpha^u_e, \eta_e, \alpha^{u'}_e \}$. Similarly, the vertex gadget corresponding to vertex $u \in V$ consists of two disjoint sets of vertices defined as $\mcal{B}_u = \{\epsilon_u, \zeta_u, \theta_u\}$ and $\mcal{S}_u = \{\kappa_u, \lambda_u, \xi_u\}$. The vertex and edge gadgets are defined for each vertex $u \in V$ and each edge $e \in E$ of the original graph $G$. The total size of each set is $|\mcal{B}| = |\mcal{S}| = 6|E| + 6 |V|$. 
	
	Formally, the following claim should be proven: $G$ has an independent set $IS(G)$ of size at least $k$ \textbf{if and only if} $G'$ admits a stable outcome with surplus at least $SW$.
	
	The transformation from an instance of \textsc{MIS} to an instance of \textsc{MBSBM} is performed in polynomial time. An important aspect of our construction is the assumption that, both in edge and vertex gadgets, there exist pairs of buyers with equal budgets. From Tables \ref{tab:edge_values} and \ref{tab:vertex_values}, one can observe that pairs of agents $\{\beta_e, \delta_e\}$ and $\{\epsilon_u, \theta_u\}$ verify the claim. Thus, the instance is not in general position, and therefore the hardness proof holds for cases where the property is violated. 
	
	We now proceed to the forward part of the proof, namely, prove that, given an independent set $IS(G)$ of size at least $k$, we construct an assignment $\mu$ as follows:
	\begin{itemize}
		\item For each $u \in V$, if $u \in IS(G)$, set $\mu(\epsilon_u) = \xi_u, \mu(\zeta_u) = \kappa_u, \mu(\theta_u) = \lambda_u$.
		\item For each $u \in V$, if $u \not\in IS(G)$, set $\mu(\epsilon_u) =\lambda_u, \mu(\zeta_u) = \xi_u,  \mu(\theta_u) = \kappa_u$.
		\item For each edge $e = \{u, u'\} \in E$ do:
		\begin{itemize}
			\item If $u \in IS(G)$, set $\mu(\beta_e) = \eta_e, \mu(\gamma_e) = \alpha^{u'}_e,  \mu(\delta_e) = \alpha^u_e$, which matches condition (1) of Lemma \ref{lem:match-conditions}. 
			\item If $u \not\in IS(G)$, set $\mu(\beta_e) = \alpha^{u'}_e, \mu(\gamma_e) = \alpha^u_e, \mu(\delta_e) = \eta_e$, which matches condition (2) of Lemma \ref{lem:match-conditions}. 
		\end{itemize}
		\item For any pair of agents $(i, j) \in G'$ not assigned above, set $\mu(i) = 0, \mu^{-1} (j) = 0$.
	\end{itemize}
	
	Since the size of the independent set is at least $k$, we can trivially verify that, according to the aforementioned rules, the surplus attained by the agents of all vertex gadgets is equal to $|V| \cdot (3 |V| + 26) + k$ and of all edge gadgets is $|E| \cdot (6|E|^2 + 27)$, thus achieving a total surplus of $SW$. In what follows, we prove that outcome $\mu$ is indeed stable.
	
	\noindent\textit{Edge Gadget.} Examining the edge gadget of edge $e = \{u, u'\} \in E$, we need to prove that there does not exist any blocking pair of agents. There are two distinct cases, as mentioned previously.
	\begin{enumerate}
		\item If $u \in IS(G)$, then $\mu$ corresponds to condition (1) of Lemma \ref{lem:match-conditions}. Therefore, no blocking pair involving agents $\{\beta_e, \gamma_e, \delta_e, \alpha^u_e, \eta_e\}$ exists. In this case, assigning $\alpha^{u'}_e$ to $\epsilon_{u'}$ would yield a higher utility for agent $\alpha^{u'}_e$. However, since $u' \not\in IS(G)$, by definition $\mu(\epsilon_{u'}) = \lambda_u$ and $\epsilon_{u'}$ is assigned to her most preferred agent, with no incentive to deviate. We conclude that there exists no blocking pair involving agents $\{\alpha^{u'}_e, \epsilon{u'}\}$, which guarantees stability for $\mu$.
		\item If $u \not\in IS(G)$, then $\mu$ corresponds to condition (2) of Lemma \ref{lem:match-conditions}. Therefore, no blocking pair involving agents $\{\beta_e, \gamma_e, \delta_e, \alpha^{u'}_e, \eta_e\}$ exists. In this case, assigning $\alpha^{u}_e$ to $\epsilon_u$ would yield a higher utility for agent $\alpha^{u}_e$. However, since $u \not\in IS(G)$, by definition $\mu(\epsilon_u) = \lambda_{u}$ and $\epsilon_u$ is integrally matched to her most preferred agent, with no incentive to deviate. We conclude that there exists no blocking pair involving agents $\{\alpha^u_e, \epsilon_u\}$, which guarantees stability for $\mu$. 
	\end{enumerate}
	
	We prove a similar result for the vertex gadgets of $G'$.

	\noindent\textit{Vertex Gadget.} Examining the vertex gadget of vertex $u \in V$, we again detect two cases.
	\begin{enumerate}
		\item If $u \in IS(G)$, then no blocking pair of $\mu$ involves agents $\{\xi_u, \theta_u, \lambda_u\}$, as $\mu$ assigns them to their most preferred agents. Therefore, no blocking pair involves agents $\{\zeta_u, \kappa_u\}$, since no possible combination can yield improved payoff for all participants. As previously argued, no agent from the edge gadget participates in a blocking pair, and thus $\epsilon_u$ is also not involved in a blocking pair.
		\item If $u \not\in IS(G)$, then no blocking pair of $\mu$ involves agents $\{\kappa_u, \zeta_u, \lambda_u\}$, as $\mu$ assigns them to their most preferred agents. The same holds for agent $\epsilon_u$. Therefore, no blocking pair involves agents $\{\theta_u, \xi_u\}$, since no possible combination can yield improved payoff for all participants.
	\end{enumerate}
	
	We conclude that neither the edge nor the vertex gadget contain agents involved in blocking pairs under $\mu$. Thus, $\mu$ is stable and achieves a surplus of $SW$.
	
	For the reverse direction, assume $\mu$ is a stable outcome for $<G', \pi_i, \pi_j>$ with surplus $surplus(\mu) \geq SW$. Defining a subset of the vertices as $IS(G) = \{u \in V\ |\ \mu(\epsilon_u) =\xi_u \}$, we prove that $IS(G)$ is an independent set of $G$ of size at least $k$. 
	
	Firstly, we prove the desired lower bound on $|IS(G)|$. We define $\mu(\epsilon_u) = N(\epsilon_u)$ as the assignment between vertex $\epsilon_u$ and any of the neighboring vertices of the edge gadgets. Summing up for all $|V|$ vertices of $G$, we get:
	\begin{align*}
		\sum_{u \in V}^{} &\ \Big( \sum_{i \in \{\epsilon_u, \zeta_u, \theta_u\}}\pi_i( \mu(i)) + \sum_{j \in \{\kappa_u, \lambda_u, \xi_u\}} \pi_j (\mu^{-1} (j)) \Big) \\\
		&\ = \sum_{u \in V}^{} (|V|+10) \cdot \mathds{1}_{ \{ \mu(\epsilon_u) = \lambda_u \}} + (|V|+12) \cdot \mathds{1}_{ \{ \mu(\epsilon_u) = \xi_u \}} + (|V|+8) \cdot \mathds{1}_{ \{ \mu(\epsilon_u) = N(\epsilon_u) \}} + \\
		&\ (|V|+7) \cdot \mathds{1}_{ \{ \mu(\zeta_u) = \kappa_u \}} + (|V|+9) \cdot \mathds{1}_{ \{ \mu(\zeta_u) = \xi_u \}} + (|V|+7) \cdot \mathds{1}_{ \{ \mu(\theta_u) = \kappa_u \}} + (|V|+8) \cdot \mathds{1}_{ \{ \mu(\theta_u) = \lambda_u \}}\\
		&\ \leq \sum_{u \in V}^{}  (|V|+10) \mathds{1}_{ \{ \mu(\epsilon_u) =\lambda_u \}} + \mathds{1}_{ \{ \mu(\epsilon_u) = \xi_u \}} + \mathds{1}_{ \{ \mu(\epsilon_u) =  N(\epsilon_u)\}} +  (|V|+8) \cdot \mathds{1}_{ \{ \mu(\zeta_u) = \kappa_u\}} + \mathds{1}_{ \{ \mu(\zeta_u) = \xi_u \}} +\\
		&\ (|V|+8) \cdot \mathds{1}_{ \{ \mu(\theta_u) = \kappa_u \}} + \mathds{1}_{ \{ \mu(\theta_u) = \lambda_u \}} + \sum_{u \in V}^{} \mathds{1}_{ \{ \mu(\epsilon_u) = \lambda_u \}} \\
		&\ \leq  |V| \cdot (3 |V|+26) + \sum_{u \in V}^{} \mathds{1}_{ \{ \mu(\epsilon_u) = \lambda_u\}}.
	\end{align*}
	
	Since $\mu$ corresponds to a binary assignment, only one condition is true for each of $\epsilon_u, \zeta_u, \theta_u$, and examining all possible outcomes, we use the mean of values to provide an upper bound on utilities. The sum of values for each vertex $\kappa_u, \lambda_u, \xi_u$ is at most 1, which is also taken into account when computing the upper bound. We can trivially verify that the total surplus of $\mu$ is maximized if and only if pair ${\epsilon_u, \xi_u}$ is integrally matched. 
	
	We now need to show that $|IS(G)| \geq k$. The set $IS(G)$ has been defined as the set of vertices ${\epsilon_u}$ for $u \in V$ that are integrally matched to ${\xi_u}$. This can be expressed as $|IS(G)| =  \sum_{u \in V}^{} \mathds{1}_{ \{ \mu(\epsilon_u) = \xi_u \}}$. It then suffices to prove inequality $\sum_{u \in V}^{} \mathds{1}_{ \{ \mu(\epsilon_u) = \xi_u \}} \geq k$, to prove the lower bound on the size of the independent set. Since the induced surplus from all $|E|$ edge gadgets is at most $surplus(\mu) \leq |E| \cdot (6|E|^2 + 27)$. At least $|V| \cdot (3 |V| +26) + \sum_{u \in V}^{} \mathds{1}_{ \{ \mu(\epsilon_u) = \xi_u \}}$ must be derived from the vertex gadgets. Using the upper bound provided above, we conclude that $\sum_{u \in V}^{} \mathds{1}_{ \{ \mu(\epsilon_u) = \xi_u \}} \geq k$, as required, thus proving the lower bound on the size of $IS(G)$. 
	
	Finally, we need to prove that $IS(G)$ is, in fact, an independent set of $G$. We prove the claim by contradiction. Suppose that there is an edge $e = \{u, u'\} \in V$  and $\{u, u'\} \subseteq IS(G)$. Then, for both nodes $u, u'$ under assignment $\mu$, it must hold that $\pi_{\epsilon_u} (\mu(\epsilon_u)) <  |V|+8 =\pi_{\epsilon_u} (\alpha^u_e)$ and $\pi_{\epsilon_{u'}} (\mu(\epsilon_{u'})) <  |V|+8 =\pi_{\epsilon_{u'}} (\alpha^{u'}_e)$, as $IS(G)$ is defined as the set of nodes $u$ that are integrally matched to $\xi_u$. From Lemma \ref{lem:match-surplus}, the induced surplus from agents of the edge gadgets is at most $6|E|^2 + 27 - |V|$, and using the previously shown bound, the surplus induced from agents of the vertex gadgets is at most $3|V|+26 + |V|$. Therefore, the total surplus of $G'$ under $\mu$ is lower than the original assumption, which is a contradiction. 
	Note that the restriction of our prices to be equal to the budget constraint $b^i$ of the buyers $i \in \mcal{B}$ in the edge and vertex gadgets is without loss of generality as shown in Lemma \ref{lem:generality}. 	
	Concluding, we have proven that indeed $IS(G)$ is an independent set of $G$.
	\qed \endproof

\end{APPENDICES}
\end{document}